\documentclass[12pt,draftclsnofoot,onecolumn]{IEEEtran}
\usepackage{graphicx}
\usepackage{amssymb}
\usepackage{amsmath}
\usepackage{cite}
\usepackage{subfigure}
\usepackage{mathrsfs}
\usepackage[displaymath,mathlines]{lineno}
\usepackage{color}
\usepackage{tabulary}
\usepackage{multirow}
\usepackage{algpseudocode}
\usepackage{algorithm,algpseudocode}
\usepackage{algorithmicx}
\usepackage{pbox}
\usepackage{multicol}
\usepackage{lipsum}
\usepackage{amsthm}
\usepackage{relsize}
\usepackage{lipsum}
\usepackage{epstopdf}
\usepackage{mathtools}
\usepackage{calc}
\usepackage{tabularx}

\makeatletter
\newcommand{\doublewidetilde}[1]{{%
		\mathpalette\double@widetilde{#1}}}
\newcommand{\double@widetilde}[2]{%
	\sbox\z@{$\m@th#1\widetilde{#2}$}%
	\ht\z@=.5\ht\z@
	\widetilde{\box\z@}}
\makeatother

\newtheorem{lemma}{Lemma}
\newtheorem{corollary}{Corollary}

\newtheorem{remark}{Remark}

\begin{document}
	\title{Model-based and Data-driven Approaches for Downlink Massive MIMO Channel Estimation}

	\author{\normalsize Amin Ghazanfari, Trinh Van Chien,   \IEEEmembership{Member,~IEEE}, Emil Bj\"{o}rnson,~\IEEEmembership{Fellow,~IEEE} and Erik G. Larsson, ~\IEEEmembership{Fellow,~IEEE} 
		\thanks{Amin Ghazanfari was with the Department of Electrical Engineering (ISY), Link\"{o}ping University (LiU), 581~83 Link\"{o}ping, Sweden and he is now with Huawei Sweden. This work was done during his PhD at LiU (email: aminghazanfari83@gmail.com).}
		\thanks{Emil Bj\"{o}rnson and Erik G. Larsson are with the Department of Electrical Engineering (ISY), Link\"{o}ping University, 581~83 Link\"{o}ping, Sweden (email: emilbjo@kth.se, erik.g.larsson@liu.se). Emil Bj\"{o}rnson is also with the Department of Computer Science, KTH Royal Institute of Technology, 164 40 Kista, Sweden.}
		\thanks{Trinh Van Chien is with the Interdisciplinary Centre for Security, Reliability and Trust (SnT), University of Luxembourg, Luxembourg City 2721, Luxembourg (email:vanchien.trinh@uni.lu).}
		\thanks{This paper was supported by ELLIIT and the Grant 2019-05068 from the Swedish Research Council. A shorter version of this work was presented at the 17th International Symposium on Wireless Communication Systems (ISWCS) \cite{ghazanfari2021learning}.}
	}

	
	\maketitle
	
	\begin{abstract}
We study downlink channel estimation in a multi-cell Massive multiple-input multiple-output (MIMO) system operating in time-division duplex. The users must know their effective channel gains to decode their received downlink data. Previous works have used the mean value as the estimate, motivated by channel hardening. However, this is associated with a performance loss in non-isotropic scattering environments.
We propose two novel estimation methods that can be applied without downlink pilots.
The first method is model-based and asymptotic arguments are utilized to identify a connection between the effective channel gain and the 
average received power during a coherence interval.
\textcolor{black}{The second method} is data-driven and trains a neural network to identify a mapping between the available information and the effective channel gain.
Both methods can be utilized for any channel distribution and precoding. {\color{black}For the model-aided method, we derive all expressions in closed form for the case when maximum ratio or zero-forcing precoding is used.}  We compare the proposed methods with the state-of-the-art using the normalized mean-squared error and spectral efficiency (SE). The results suggest that the two proposed methods provide better SE than the state-of-the-art when there is a low level of channel hardening, while the performance difference is relatively small with the uncorrelated channel model.
	\end{abstract}

	\IEEEpeerreviewmaketitle

	\section{Introduction}
	Massive multiple-input multiple-output (MIMO) is one of the backbone technologies for 5G-and-beyond networks \cite{larsson2014massive,Parkvall2017a}. In Massive MIMO, each base station (BS) is equipped with many active antennas to facilitate adaptive beamforming towards individual users and spatial multiplexing of many users \cite{marzetta2010noncooperative}. In this way, the technology can improve the spectral efficiency (SE) for individual users and, particularly, increase the sum SE in highly loaded networks by orders of magnitude compared with conventional cellular technology with passive antennas \cite{bjornson2017massive}. Having accurate channel state information (CSI) is essential in Massive MIMO networks \cite{larsson2014massive}, so that the transmission and reception can be tuned to the user channels, to amplify desired signals and reject interference.
	Time-division duplex (TDD) operation is preferable for CSI acquisition because the BSs can then acquire uplink CSI from the uplink pilot transmission and utilize the uplink-downlink channel reciprocity to transform it to downlink CSI \cite{bjornson2017massive}. In this way, the required pilot resources are proportional to the number of users but independent of the number of BS antennas. In contrast, in frequency-division duplex (FDD), the BSs must send downlink pilots to help the users to estimate the downlink channels, which requires pilot resources are proportional to the number of BS antennas. Then, the users must feed back the estimated downlink channels to the BSs, which consumes further resources  \cite{bjornson2017massive}.
	Thus, in this paper, we consider a TDD Massive MIMO network, which is in line with the current 5G deployments that mainly utilize TDD bands.
	
	To decode the downlink signals coherently, each user must provide an estimate of the effective downlink channel gain and the effective noise variance to the decoding algorithm.
	The effective downlink channel gain is an inner product of the precoding vector and the channel vector, thus the user only needs to know this scalar, not the individual vectors. The value of this scalar varies due to channel fading. 
	A common approach in the Massive MIMO literature is to estimate the scalar based on its long-term statistics \cite{jose2011pilot,yang2013performance,redbook}, which are constant and thus can be assumed to be known at the user.
	In particular, the mean of the effective channel gain can be used as the estimate, motivated by the channel hardening effect, which says that the variations around the mean are small when the number of antennas is large \cite{redbook}. However, the number of antennas that is needed to observe channel hardening depends strongly on the propagation environment. There are also experimental papers that quantify channel hardening, e.g. \cite{gunnarsson2020channel}, where the authors show that assumption of uncorrelated Rayleigh fading is overoptimistic and Massive MIMO channels are actually spatially correlated.  
	\textcolor{black}{With spatially correlated fading, one might need hundreds of antennas to achieve the same channel hardening level as uncorrelated Rayleigh fading \cite[Fig.~2.7]{bjornson2017massive}.} 
	Estimating the effective downlink channel gain using the mean value will result in a significant SE loss when the channel hardening level is low (or non-existent) \cite{ngo2017no}. 
	
	Another approach is that the BSs transmit downlink pilots to assist the users in estimating the effective downlink channel gains \cite{ngo2013massive}. These pilots can be beamformed using the same precoding as the downlink data, thus the required downlink pilot 
resources are proportional to the number of users, but will anyway increase the total overhead for CSI acquisition. \textcolor{black}{While the use of downlink pilots will improve the estimates of the effective downlink channel gains when the number of users is small or significantly less than the dimensions of the coherence intervals, the SE might decrease due to the extra overhead when the system simultaneously serves many users  \cite{ngo2017no}. Downlink pilots is particularly cumbersome in multi-cell systems since there will be pilot contamination and the number of pilots must be dimensioned for the cells with the largest number of active users.}

	There exists one more approach, namely blind channel estimation \cite{ngo2015blind}. In this approach, the effective downlink channel gains are estimated using the downlink data signals. The enabling factor is that the precoding is selected to make the effective downlink channel gains (approximately) positive and real-valued, so that only the amplitude must be estimated.  \textcolor{black}{ An estimation method of this kind was developed for single-cell Massive MIMO systems in	\cite{ngo2017no,ngo2015blind}, but without considering pilot contamination and spatial correlation.}	
 
	The method was motivated by asymptotic arguments that are generally not satisfied in Massive MIMO, but was shown to provide better estimates than when using the mean of the effective channel gains. The difference was substantial in scenarios without channel hardening. The method also performed better than the use of downlink pilots since the blind channel estimation method does not need any extra pilot overhead. The blind estimation method was generalized in  \cite{pasangi2020blind} to  a multi-cell Massive MIMO network with uncorrelated Rayleigh fading channel and maximum ratio (MR) precoding at the BS. {\color{black} In this paper, we propose a general multi-cell framework that supports correlated Rayleigh fading, arbitrary linear precoding, and generic pilot assignment among the users.}

Blind estimation builds on identifying a mapping between received data signals and the variable that is to be estimated. This is an intrinsically complicated problem and the man-made solutions often rely on asymptotic arguments, where parameters in the system model are taken towards infinity to untangle the relationships \cite{muller2014blind,amiri2017blind,ghavami2017blind,shin2007blind,de1997asymptotic,de1997cramer}. There is no guarantee that such algorithms will yield optimal estimation quality.
	The deep learning methodology has the potential to find mappings between input data and desired variables by learning them from data \cite{o2017introduction,van2020power,sanguinetti2018deep}.
	Data-driven approaches based on deep learning algorithms have been proposed to solve various physical layer and resource allocation problems in wireless communications \cite{o2017introduction,sadeghi2018adversarial}. For example, neural networks have been trained to achieve efficient interference management \cite{Sun2018a}, power control \cite{sanguinetti2018deep}, and channel estimation \cite{demir2019channel}. 
	Deep learning methods can be particularly useful to solve problems for which 1) we do not have good physical models that enable model-based algorithmic development or 2) we have good algorithms but these are computationally intractable  \cite{bjornson2020two}. For example, one application of deep learning can be the real-time implementation of high computational complexity iterative algorithms \cite{Balatsoukas2019}.
	The drawback of neural networks is the complicated training process and limited explainability, thus neural networks might only be practically useful in situations where the performance gains are substantial.
	\textcolor{black}{In this paper, we use neural networks to develop a data-driven blind channel estimator of the downlink effective scalar channel gain in cellular TDD Massive MIMO systems. It is the first time the data-driven approach is investigated for this purpose. However, a few previous works have dealt with the downlink channel estimation for FDD systems \cite{9347820}, which is a fundamentally different problem since the entire multi-antenna channel is to be estimated (to enable feedback and precoding computation) and not the effective channel gain after precoding has been applied.}

	The main contributions of this manuscript are summarized as follows:
	\begin{itemize}
	
		\item \textcolor{black}{We propose a framework for blind estimation of the downlink effective channel gain in multi-cell Massive MIMO networks, which supports an arbitrary channel model spanning from non-isotropic to isotropic scattering environments, arbitrary precoding, and different pilot reuse factors.}
		
		\item We propose a model-aided estimator that utilizes channel statistics. The estimator is motivated by asymptotic arguments and the asymptotic performance is evaluated analytically. We compute all terms in closed form
		when using MR precoding along with correlated Rayleigh fading and when using zero-forcing (ZF) precoding along with uncorrelated Rayleigh fading. An SE expression is obtained to quantify the achievable downlink performance.

		\item {\color{black} We propose a data-driven blind estimator for finite-sized systems based on a fully-connected neural network.
		We train the neural network to estimate the downlink effective channel gains from input features such as the pathloss, propagation environment, and average power of the received signal at the user.}
		
		\item We provide numerical results that manifest the proposed estimation methods' effectiveness in terms of both the normalized mean square error (NMSE) and the downlink SE. We evaluate under what channel conditions  the gains are large compared to the state-of-the-art.
			\end{itemize}
	
The remainder of this paper is organized as the following: Section \ref{sec:system-model} presents the system model for multi-cell massive MIMO network along with the uplink pilot and downlink data transmissions analysis. Section \ref{sec:modelbased}, presents a model-based approach for estimation of downlink effective channel gains. In Section \ref{sec:SE}, we present the analysis of ergodic spectral efficiency for downlink data transmission and asymptotic analysis of model-based estimator. Section \ref{sec:datadriven} is dedicated to present a data-driven approach for estimation of downlink effective channel gains. Numerical results are presented and analyzed in Section \ref{sec:numerical-results}. Finally, the main conclusions of this paper are presented in Section \ref{sec:conclusion}.

	\textbf{Notation:} In this paper, we use boldface lower case to indicate column vectors $\mathbf{x}$ and boldface upper case is used for matrices, $\mathbf{X}$. An identity matrix with size $M$ is denoted as $\mathbf{I}_M$. The conjugate transpose of $\mathbf{X}$ is denoted as $\mathbf{X}^{{\rm H}}$.  Furthermore, the operators $\mathbb{E}\{\boldsymbol{\cdot}\}$ and $\mathrm{var}\{\boldsymbol{\cdot}\}$ denote the expectation and variance of a random variable, respectively. The notation $\|\mathbf{x}\|$ stands for the Euclidean 
	norm of the vector $\mathbf{x}$. The notation $ \mathcal{CN}({\mathbf{0}},{\mathbf{R}})$ is used for the (multi-variate) circularly symmetric complex Gaussian distribution with correlation matrix ${\mathbf{R}}$. Convergence in probability is denoted as $\xrightarrow{P}$. The corresponding asymptotic equivalence between two functions $f(K)$ and $g(K)$ is denoted as
$f(K) \asymp  g(K)$, and holds if $f(K) -  g(K) \xrightarrow{P} 0$ when  $K\rightarrow \infty$.

	\section{System Model}\label{sec:system-model}
	In this paper, we consider a multi-cell Massive MIMO system consisting of $L$ cells. Each of the cells has a BS equipped with $M$ antennas and serves $K$ single-antenna users. The wireless propagation channels vary over time and frequency, which is modeled by the standard block-fading channel model \cite[Sec.~2]{redbook}. In this model, the time-frequency resources are divided into coherence intervals where all the channels are (approximately) static and frequency flat. The size of a coherence interval (in number of transmission symbols) is denoted $\tau_c$.
	The channels change independently from one coherence interval to the next according to a stationary ergodic random process.  {\color{black} The channel between BS~$l$ and user~$k'$ in cell~$l'$, $\mathbf{g}_{l'k'}^l \in \mathbb{C}^M$, which follows a correlated Rayleigh fading model:
		\begin{equation}
		\mathbf{g}_{l'k'}^l \sim \mathcal{CN}\left({\mathbf{0}}, \mathbf{R}_{l'k'}^l \right),
		\end{equation}
		where ${\mathbf{R}}^{l}_{l'k'} \in \mathbb{C}^{M\times M}$ is the positive semi-definite spatial correlation matrix of the channel.\footnote{\textcolor{black}{We assume that the spatial correlation matrices are known in this paper for the sake of simplicity. In practical systems, the correlation matrices can be estimated by averaging  over many different instantaneous channel realizations. We refer to \cite{Sanguinetti2019az} for a recent review of such estimation methods.} } We define $\beta^{l}_{l'k'}  = \frac{\mathrm{tr}\left(\mathbf{R}^{l}_{l'k'}\right)}{M}$, where $\beta^{l}_{l'k'} \geq 0$ is the corresponding average large-scale fading coefficient among the $M$ antennas. The elements of $\mathbf{R}^{l}_{l'k'}$ are also determined by the array geometry and angular distribution of the multipath components \cite{bjornson2017massive}. A special case is independent and identically distributed (i.i.d.) Rayleigh fading with $\mathbf{R}^{l}_{l'k'} = \beta^{l}_{l'k'} \mathbf{I}_M$, which can be obtained in isotropic scattering environments but seldom occur in practice.}
	 
	We focus on the downlink data transmission in a network operating with a TDD protocol. Since a new independent channel realization appears in every coherence interval, the BSs need to estimate them once per coherence interval.
	More precisely, each BS estimates the channels based on pilots transmitted during the uplink training phase and prior knowledge of the channel statistics. This procedure is explained in the next subsection. 
	
	\subsection{Uplink Pilot Training}
	To enable spatial multiplexing in the downlink, the BS must estimate the channel to the intra-cell users in every coherence interval and construct precoding vectors based on them. {\color{black} In Massive MIMO, this is achieved by letting the users transmit uplink pilots since an arbitrarily large number of BS antennas can then be supported. If the network contains many cells, the same pilot sequences may have to be reused in several cells. This results in pilot contamination, which we will combat by exploiting the user-specific spatial correlation
through MMSE estimation \cite{BjornsonHS17} and a pilot reuse factor of $f \geq 1$. The latter means the users within each cell utilize mutually orthogonal pilots, and the same pilot sequences are reused in a fraction $1/f$ of the $L$ cells in the network.}

	To achieve this, we assume there is a set of $\tau_p = fK $ mutually orthogonal pilot sequences, each of length $\tau_p$. 
	The pilot sequence used by user $k$ in cell $l$ is denoted as $\pmb{\phi}_{lk} \in \mathbb{C}^{\tau_p}$ and satisfies $\| \pmb{\phi}_{lk} \|^2 = \tau_p$. 
	We let $\mathcal{P}_{l}$ denote the set of cells sharing the same subset of $K$ orthogonal pilot sequences as cell $l$. When two cells use the same subset of pilots, the users are numbered so that user $k$ in both cells use the same pilot. Hence, it follows that the inner product of two users' pilot sequences are
		\begin{equation}
	\pmb{\phi}^{\rm H}_{lk} \pmb{\phi}_{l'k'} = \begin{cases}
	\tau_p, & \mbox{if } k = k'~\text{and}~l' \in \mathcal{P}_{l} , \\
	0,  & \mbox{otherwise}.
	\end{cases}
	\end{equation} 
	The received uplink signal $\mathbf{Y}_l \in \mathbb{C}^{M \times \tau_p }$ at BS $l$ during the uplink pilot transmission is
	\begin{equation}
	\mathbf{Y}_l = \sum_{l'=1}^L \sum_{k'=1}^K \sqrt{\hat{p}_{l'k'}} \mathbf{g}_{l'k'}^l \pmb{\phi}_{l'k'}^{\rm H} + \mathbf{N}_l, 
	\end{equation}
	where $\hat{p}_{l'k'}$ is the pilot power used by user $k'$ in cell $l'$ and the additive noise matrix $\mathbf{N}_l \in \mathbb{C}^{M \times \tau_p }$ with independent  complex Gaussian distributed entries: $\mathcal{CN}(0, \sigma_{\mathrm{UL}}^2)$. 
	A first step when BS $l$ wants to estimate the channel of user $k$ in cell $l$ is to correlate its received signal with the user's pilot sequence:
	\begin{equation} \label{eq:Observe1}
	\tilde{\mathbf{y}}_{lk} =  \mathbf{Y}_l \pmb{\phi}_{lk} =\tau_p \sqrt{\hat{p}_{lk}} \mathbf{g}_{lk}^l + \sum\limits_{l' \in \mathcal{P}_{l} \setminus \{ l \}}  \tau_p \sqrt{\hat{p}_{l'k}} \mathbf{g}_{l'k}^l + \tilde{\mathbf{n}}_{lk},
	\end{equation}
	where $\tilde{\mathbf{n}}_{lk} = \mathbf{N}_l \pmb{\phi}_{lk} \sim \mathcal{CN}(\mathbf{0}, \tau_p \sigma_{\mathrm{UL}}^2 \mathbf{I}_M )$ is independent of the cell and user indices. The MMSE channel estimate of $\mathbf{g}_{lk}^l$ is presented in the following lemma.
	
	\begin{lemma} \label{lemma:ChannelEstCorrelated}
		When BS~$l$ uses MMSE estimation based on the observation \eqref{eq:Observe1}, the estimate of the channel between user~$k$ in cell~$l$ is
		\begin{equation}
		\hat{\mathbf{g}}_{lk}^l = \sqrt{\hat{p}_{lk}} \mathbf{R}_{lk}^l \pmb{\Psi}_{lk}^{-1} \tilde{\mathbf{y}}_{lk},
		\end{equation}
		where $\pmb{\Psi}_{lk} = \sum\limits_{l'\in \mathcal{P}_l} \tau_p  \hat{p}_{l'k} \mathbf{R}_{l'k}^l + \sigma_{\mathrm{UL}}^2 \mathbf{I}_M $ and $\hat{\mathbf{g}}_{lk}^l$ is distributed as $\hat{\mathbf{g}}_{lk}^l \sim \mathcal{CN} \left(\mathbf{0}, \tau_p \hat{p}_{lk} \mathbf{R}_{lk}^l \pmb{\Psi}_{lk}^{-1} \mathbf{R}_{lk}^l \right)$. Meanwhile, the channel estimation error  $\mathbf{e}_{lk}^l = \mathbf{g}_{lk}^l - \hat{\mathbf{g}}_{lk}^l$ is independently distributed as $\mathbf{e}_{lk}^l \sim \mathcal{CN} \left( \mathbf{0}, \mathbf{C}_{lk}^l \right)$, where we define $\mathbf{C}_{lk}^l = \mathbf{R}_{lk}^l - \tau_p \hat{p}_{lk} \mathbf{R}_{lk}^l \pmb{\Psi}_{lk}^{-1} \mathbf{R}_{lk}^l$.
	\end{lemma}
	\begin{proof}
		The results are obtained by adopting the standard MMSE estimation \cite{kay1993fundamentals} to our system model. The detailed derivation and proof can be find in \cite[Theorem~3.1]{bjornson2017massive}.
	\end{proof}
	
	Lemma~\ref{lemma:ChannelEstCorrelated} gives the exact expression of the channel estimates together with the statistical information. 
When discussing the new analytical results and algorithms, we will sometimes consider the i.i.d.~Rayleigh fading case with ${\mathbf{R}}^{l}_{l'k'} = \beta^{l}_{l'k'} {\mathbf I}_{M}$ to obtain closed-form expressions. In this case, the MMSE estimator in Lemma~\ref{lemma:ChannelEstCorrelated} simplifies as follows.

	\begin{corollary} \label{lemma:ChannelEst}
	If ${\mathbf{R}}^{l}_{l'k'} = \beta^{l}_{l'k'} {\mathbf I}_{M}$, the MMSE estimate of channel between user~$k$ in cell~$l$ and BS~$l$ is
		\begin{equation} \label{eq:chRay}
		\hat{\mathbf{g}}_{lk}^l = \frac{\sqrt{\hat{p}_{lk}} \beta_{lk}^l }{ \tau_p \sum\limits_{l' \in \mathcal{P}_l} \hat{p}_{l'k} \beta_{l'k}^l + \sigma_{\mathrm{UL}}^2 } \tilde{\mathbf{y}}_{lk},
		\end{equation}
		and is distributed as $\hat{\mathbf{g}}_{lk}^l \sim \mathcal{CN}(\mathbf{0}, \gamma_{lk}^l \mathbf{I}_M)$ with the variance
		\begin{equation} \label{eq:gammalkl}
		\gamma_{lk}^l = \frac{\tau_p \hat{p}_{lk} \left(\beta_{lk}^l\right)^2 }{ \tau_p \sum\limits_{l'\in \mathcal{P}_l} \hat{p}_{l'k} \beta_{l'k}^l + \sigma_{\mathrm{UL}}^2 }.
		\end{equation}
		The channel estimation error is  
		distributed as $\mathbf{e}_{lk}^l \sim \mathcal{CN}(\mathbf{0}, (\beta_{lk}^l - \gamma_{lk}^l) \mathbf{I}_M)$. 
	\end{corollary}

Pilot contamination occurs when two users use the same pilot sequence. The mutual interference leads to increased estimation errors but also an inability to separate the users' channels. The latter is particularly clear for i.i.d.~Rayleigh fading because the MMSE estimates of the channels of two pilot-sharing users are equal up to a scaling factor:
	\begin{equation} \label{eq:ChannelEstRelation}
	\hat{\mathbf{g}}_{l'k}^l = \frac{\sqrt{\hat{p}_{l'k} } \beta_{l'k}^l  }{\sqrt{\hat{p}_{lk} } \beta_{lk}^l } \hat{\mathbf{g}}_{lk}^l,
	\end{equation}
where $\hat{\mathbf{g}}_{lk}^l$ is defined in \eqref{eq:chRay} and both $l$ and $l'$ cells are sharing the same set of pilots.

\subsection{Downlink Data Transmission}
The channel estimates and their statistics are used to formulate the downlink precoding vectors. Since we focus on the downlink transmission, we assume that the remaining $\tau_{c} -\tau_{p}$ transmission symbols per coherence interval are used for downlink data. Let us denote by $s_{lk}[n]$ the $n$-th data symbol that BS $l$ transmits to user~$k$ in cell~$l$, where  $n$ is an index from $1$ to $\tau_{c} -\tau_{p}$. The data symbols have zero mean and normalized power: $\mathbb{E} \{ |s_{lk}[n]|^2 \} = 1$. By assigning a linear precoding vector $\mathbf{w}_{lk} \in \mathbb{C}^M$ to the user, the signal that BS~$l$ sends to the $K$ users in cell~$l$ is
	\begin{equation}
	\mathbf{x}_{l}[n] = \sum_{k=1}^K \sqrt{\rho_{\rm dl}\eta_{lk}} \mathbf{w}_{lk}s_{lk}[n],
	\end{equation}
	where $\rho_{\rm dl}$ is the maximum downlink transmit power and $\eta_{lk} \in [0,1]$ is power allocation coefficient that determines the power fraction assigned to user~$k$ in cell~$l$. We consider an arbitrary fixed selection $\eta_{l1},\ldots,\eta_{lK}$ in every cell but note that it must be selected such that $\sum_{k=1}^{K} \eta_{lk} \leq 1$ that should hold for all $l$ cells, to comply with the maximum downlink transmit power limit. 
	
	The received signal $y_{lk} [n]$ at user~$k$ in cell~$l$ is
	\begin{equation} \label{eq:ReceivedSig}
	\begin{aligned}
	y_{lk} [n] &= \sum_{l'=1}^L \left(\mathbf{g}_{lk}^{l'}\right)^{\rm H} \mathbf{x}_{l'} [n] + \tilde{w}_{lk}[n]  \\
	&= \sqrt{\rho_{\rm dl}\eta_{lk}} \left(\mathbf{g}_{lk}^{l}\right)^{\rm H} \mathbf{w}_{lk}s_{lk}[n] + \sum\limits_{\substack{k'=1, k' \neq k}}^K  \sqrt{\rho_{\rm dl}\eta_{lk'}} \left(\mathbf{g}_{lk}^{l}\right)^{\rm H} \mathbf{w}_{lk'}s_{lk'}[n] \\
	&+ \sum\limits_{\substack{l'=1,l' \neq l}}^L \sum_{k'=1}^K \sqrt{\rho_{\rm dl}\eta_{l'k'}} \left(\mathbf{g}_{lk}^{l'}\right)^{\rm H} \mathbf{w}_{l'k'}s_{l'k'}[n]  + \tilde{w}_{lk}[n],
	\end{aligned}
	\end{equation}
	where $\tilde{w}_{lk}[n] \sim \mathcal{CN}(0, \sigma_{\mathrm{DL}}^2)$ is the additive noise. In the last expression of \eqref{eq:ReceivedSig}, the first term is the desired signal for user~$k$ in cell~$l$ and the second term is the intra-cell interference. The remaining terms are the inter-cell interference and noise. Each of the signal and interference terms contain a data symbol multiplied with an effective downlink channel gain of the type
	\begin{equation} \label{eq:alphaVal}
	\alpha_{lk}^{l'k'} = \sqrt{\rho_{\rm dl}} \left(\mathbf{g}_{lk}^{l'}\right)^{\rm H} \mathbf{w}_{l'k'}.
	\end{equation}
	Using this notation, 
	we can then rewrite \eqref{eq:ReceivedSig}  as
	\begin{equation}\label{eq:ReceivedSigv1}
	y_{lk} [n] = \sqrt{\eta_{lk}}\alpha_{lk}^{lk} s_{lk}[n] + \sum\limits_{\substack{k'=1, k' \neq k}}^K  \sqrt{\eta_{lk'}}\alpha_{lk}^{lk'} s_{lk'}[n]
	+ \sum\limits_{\substack{l'=1, l' \neq l}}^L \sum_{k'=1}^K \sqrt{\eta_{l'k'}}\alpha_{lk}^{l'k'} s_{l'k'}[n]  + \tilde{w}_{lk}[n].
	\end{equation}
	To decode the desired signal $s_{lk}[n]$, user~$k$ in cell~$l$ should preferably know the effective channel gain $\alpha_{lk}^{lk}$
	and the average power of the remaining interference-plus-noise terms.
	Learning the effective channel gain $\alpha_{lk}^{lk}$ is the most critical issue since its value changes in every coherence interval, thus an efficient downlink channel estimation procedure is needed.
	One option is to spend a part of the coherence interval on transmitting downlink pilots \cite{ngo2013massive}. 
	Another option is to utilize the structure created by the fact that the precoding vector is computed based on an MMSE estimate of $ \mathbf{g}_{lk}^{l} $. Although $\mathbb{E}\{ \mathbf{g}_{lk}^{l} \} =\mathbf{0}$, we have $\mathbb{E}\{ \alpha_{lk}^{lk} \} > 0$ for most precoding schemes, thus a basic estimate of $\alpha_{lk}^{lk}$ is its mean value $\mathbb{E}\{ \alpha_{lk}^{lk} \}$ \cite{Marzetta2006a}.
	The latter solution is attractive in Massive MIMO systems where the channel hardening property implies that $\alpha_{lk}^{lk}$ is relatively close to its mean value \cite{jose2011pilot,Marzetta2006a}.
The drawback with these solutions are the extra pilot overhead and the substantial performance reduction in the high-SNR regime, respectively. Consequently, this paper develops improved methods for downlink channel estimation which are blind (i.e., does not require {\color{black} explicit downlink  pilot signals}) and are either model-aided or designed using deep learning.

	\section{Model-based Estimation of the Effective Downlink Channel Gain}\label{sec:modelbased}
	
	We want to estimate the realization of the effective downlink channel gain $\alpha_{lk}^{lk} = \sqrt{\rho_{\rm dl}} \left(\mathbf{g}_{lk}^{l}\right)^{\rm H} \mathbf{w}_{lk}$ in \eqref{eq:alphaVal}, for each user $k$ in a given cell $l$, without transmitting explicit downlink pilots. To this end, we notice that many precoding schemes make use of precoding vectors of the type 
	\begin{equation} \label{eq:precoding-structure}
	\mathbf{w}_{lk} = \mathbf{A}_{lk} \hat{\mathbf{g}}_{lk}^{l},
	\end{equation}
which involves a positive definite matrix $\mathbf{A}_{lk}$ and the channel estimate $ \hat{\mathbf{g}}_{lk}^{l}$ \cite{bjornson2017massive}. For example, MR precoding is obtained when $\mathbf{A}_{lk}$ is a scaled identity matrix.
For precoding schemes of this type, the effective downlink channel gain becomes
\begin{equation}
\alpha_{lk}^{lk} = \sqrt{\rho_{\rm dl}} \left(\mathbf{g}_{lk}^{l}\right)^{\rm H}  \mathbf{A}_{lk} \hat{\mathbf{g}}_{lk}^{l} =  \underbrace{\sqrt{\rho_{\rm dl}} \left(\hat{\mathbf{g}}_{lk}^{l}\right)^{\rm H}  \mathbf{A}_{lk} \hat{\mathbf{g}}_{lk}^{l}}_{>0} +  \underbrace{\sqrt{\rho_{\rm dl}} \left(\mathbf{e}_{lk}^{l}\right)^{\rm H}  \mathbf{A}_{lk} \hat{\mathbf{g}}_{lk}^{l}}_{\approx 0},
\end{equation}
where the second term is almost zero (relative to the first term) since the estimation error $\mathbf{e}_{lk}^{l}$ is independent of the estimate, and both have zero mean. This is particularly true when having many antennas due to the effect known as channel hardening.
Hence, even if the user does not know $\alpha_{lk}^{lk}$, it knows that it should estimate an approximately positive real-valued number. This can be done by inspecting the average power of the received downlink signals in the coherence interval, without the need for explicit pilots \cite{ngo2017no}. Note that downlink pilots would have been needed if there was also an unknown phase in  $\alpha_{lk}^{lk}$ but it is alleviated by the coherent precoding. {\color{black}In this section, we will develop such a blind channel estimation by extending the single-cell scheme from \cite{ngo2017no}  to cellular Massive MIMO communications consists of multiple cells.}

	\subsection{Proposed Blind Channel Estimation with an Arbitrary Precoding Technique}
	
	We will outline and motivate the proposed scheme for estimating the effective channel gain $\alpha_{lk}^{lk}$ of user~$k$ in cell~$l$. 
	As a first step, the user can compute the arithmetic sample mean of the received signal power in the current coherence interval:
	\begin{equation} \label{eq:Xilko}
	\xi_{lk}  = \frac{\sum_{n=1}^{\tau_c - \tau_p} \left|y_{lk} [n]\right|^2 }{\tau_c - \tau_p}.
	\end{equation}
	The data signals and noise take new realizations for every $n$, thus we can expect that these sources of randomness will be averaged out in $\xi_{lk}$ when the length of the coherence interval is large. This result can be formalized mathematically as follows.

	\begin{lemma}\label{Lemma:asymptotic}
		As $\tau_c \rightarrow \infty$ (for a fixed $ \tau_p $), $\xi_{lk}$ in \eqref{eq:Xilko} converges in probability as follows:
		\begin{equation} \label{eq:Xilk}
		\xi_{lk} \xrightarrow{P} \left( \eta_{lk}\left|\alpha_{lk}^{lk}\right|^2 + \sum\limits_{\substack{k'=1,\\ k' \neq k}}^K  \eta_{lk'}\left|\alpha_{lk}^{lk'}\right|^2 + \sum\limits_{\substack{l'=1,\\ l' \neq l}}^L \sum_{k'=1}^K \eta_{l'k'}\left|\alpha_{lk}^{l'k'}\right|^2 + \sigma_{\mathrm{DL}}^2  \right).
		\end{equation}
	\end{lemma}
	\begin{proof}
		The proof is based on inserting \eqref{eq:ReceivedSigv1} into \eqref{eq:Xilko} and then using the law-of-large-numbers and the mutual independence between signals and noise. The proof is given in Appendix~\ref{Appendix:asymptotic}.
	\end{proof}
	The first term at the right-hand side of \eqref{eq:Xilk} contains the desired channel gain of user~$k$ in cell~$l$, while the second term represents the intra-cell interference. The third term is the inter-cell interference and the last term is the noise variance. Note that the right-hand side of \eqref{eq:Xilk} is constant within a coherence interval but takes different independent realizations in different blocks.	
	Hence, the convergence in probability in \eqref{eq:Xilk} refers to the randomness of the signals and noise, but is conditioned on the channel realizations in the considered coherence interval.
	Our goal is to utilize the asymptotic limit in \eqref{eq:Xilk} to estimate $\alpha_{lk}^{lk}$  from $\xi_{lk}$, but this is an ill-posed estimation problem since there are $LK$ unknowns: $\alpha_{lk}^{l'k'}$, $l'=1,\ldots,L$, $k'=1,\ldots,K$.
	To resolve this issue, we will make use of another asymptotic result, based on the regime where the number of users is large.

	\begin{lemma} \label{lemma:asymptoicLLN}
	Suppose the users are dropped in each cell independently at random according to some common distribution for which $\alpha_{lk}^{l'k'}$ has bounded variance.
	As $\tau_c,K \to \infty$ such that $K/\tau_c \to 0$ and $\tau_p = fK$, we obtain the following asymptotic equivalence:  		
		\begin{equation}\label{eq:aymptotLemma}
		\frac{1}{K }\xi_{lk}  \asymp \frac{1}{K} \left(\eta_{lk}\left|\alpha_{lk}^{lk}\right|^2 + \sum\limits_{\substack{k'=1,\\ k' \neq k}}^K  \eta_{lk'} \mathbb{E} \left\{ \left|\alpha_{lk}^{lk'}\right|^2 \right\}
		+ \sum\limits_{\substack{l'=1,\\ l' \neq l}}^L \sum\limits_{k'=1}^K \eta_{lk'}  \mathbb{E} \left\{ \left|\alpha_{lk}^{l'k'}\right|^2 \right\} \\
		+ \sigma_{\mathrm{DL}}^2\right).   
		\end{equation}
	\end{lemma}
	\begin{proof}  
		This result follows by reformulating the expression and then applying the law-of-large-numbers. The proof is given in Appendix~\ref{Appendix:asymptoicLLN}.
	\end{proof}	
	Lemma \ref{lemma:asymptoicLLN} indicates that we can replace the mutual interference terms by their mean values as $K \rightarrow \infty$. Note that the mean value is computed with respect to the channel realizations for given user locations. 
	{\color{black} The lemma is a rigorous asymptotic result that we will utilize it as a motivation for approximating $\xi_{lk}$ for a finite number of users $K$ per cell as 
		\begin{equation} \label{eq:Xilkv1Orig}
		\xi_{lk}  \approx \left(\eta_{lk}\left|\alpha_{lk}^{lk}\right|^2 + \sum\limits_{\substack{k'=1,\\ k' \neq k}}^K  \eta_{lk'} \mathbb{E} \left\{ \left|\alpha_{lk}^{lk'}\right|^2 \right\}
		+  \sum\limits_{\substack{l'=1,\\ l' \neq l}}^L \sum\limits_{k'=1}^K  \eta_{l'k'} \mathbb{E} \left\{ \left|\alpha_{lk}^{l'k'}\right|^2 \right\} \\
		+ \sigma_{\mathrm{DL}}^2\right).
		\end{equation}
		The approximation in \eqref{eq:Xilkv1Orig} is expected to perform well for practical systems with a limited number of users.}\footnote{{\color{black} There is no theoretical guarantee on the estimation performance if the asymptotic conditions in Lemma~\ref{lemma:asymptoicLLN} are not satisfied, but good approximation for a finite number of users in the system is expected as a rule of thumb.}}

	If there would have been equality in this expression, we can solve for $|\alpha_{lk}^{lk}|$ and obtain the following estimator which is the same estimator as provided in \cite{ngo2017no} for the single-cell case. 
	\begin{equation}
	\alpha_{lk}^{lk} \approx |\alpha_{lk}^{lk}| \approx \sqrt{\frac{ \xi_{lk}  - T_{lk}  }{\eta_{lk}}}
	\end{equation}
	by utilizing that $\alpha_{lk}^{lk}$ is approximately positive and by defining the following variable:
	\begin{equation}\label{eq:TLK}
	T_{lk} = \sum_{k'=1, k' \neq k}^K  \eta_{lk'}  \mathbb{E} \left\{ \left|\alpha_{lk}^{lk'}\right|^2 \right\} + \sum_{l'=1, l' \neq l}^L \sum_{k'=1}^K \eta_{l'k'} \mathbb{E} \left\{ \left|\alpha_{lk}^{l'k'}\right|^2 \right\} + \sigma_{\mathrm{DL}}^2.
	\end{equation}
		Based on \eqref{eq:Xilkv1Orig}, we propose the following estimator of the effective channel gain:
	\begin{equation} \label{eq:alphalkhat}
	\hat{\alpha}_{lk}^{lk} = \begin{cases}
	\sqrt{\frac{ \xi_{lk}  - T_{lk}  }{\eta_{lk}}}, & \mbox{if } \xi_{lk} > \Theta_{lk} ,  \\
	\mathbb{E} \big\{ \alpha_{lk}^{lk} \big\},& \mbox{otherwise}.
	\end{cases}
	\end{equation}
	The second case corresponds to utilizing the mean value as the estimate of the effective channel gain, as previously proposed in \cite{Marzetta2006a}. The switching point $\Theta_{lk}$ between the two cases in \eqref{eq:alphalkhat} should have a value above $T_{lk}$, so that $ \xi_{lk}  - T_{lk} > 0$ and thus the square-root provides a real number. We will show later that somewhat larger switching points are desirable to maximize the performance of the proposed estimator.
	{\color{black} We stress that the proposed estimator in \eqref{eq:alphalkhat} can be applied along with any precoding technique, but then the expectations with respect to the small-scale fading will have to be computed by numerical integration or Monte-Carlo methods. The latter can be utilized in practice by taking a sample average based on measurements obtained in different coherence intervals in time and frequency.}
	 Next, we will provide examples of when the expectations can be computed in closed form. To compare the performance of our proposed approach with other available schemes, we consider the normalized MSE at the user $k$ in cell $l$ defined as
	\begin{equation} \label{eq:NMSE}
	\text{MSE}_{lk} = \frac{\mathbb{E}\{|\hat{\alpha}_{lk}^{lk}-{\alpha}_{lk}^{lk}|^2\}}{\mathbb{E}\{|{\alpha}_{lk}^{lk}|^2\}}.
	\end{equation}

	\subsection{Downlink Channel Estimation With Linear Precoding Techniques}
	
	{\color{black} We will now consider a few different linear precoding techniques, motivated by the fact that non-linear precoding has negligible spectral efficiency benefits in Massive MIMO and the expectations in  \eqref{eq:alphalkhat} can be computed in closed form for linear precoding.} We begin with MR precoding, in which case	
	\begin{equation} \label{eq:PrecodVecCorr}
	\mathbf{w}_{lk} = \frac{\hat{\mathbf{g}}^{l}_{lk}}{\sqrt{\mathbb{E}\left\{\|\hat{\mathbf{g}}^{l}_{lk} \|^{2}\right\}}}
	\end{equation}
	where the normalization term makes sure that $\mathbb{E} \{ \| \mathbf{w}_{lk} \|^2 \} = 1$ \cite{redbook}.\footnote{{\color{black}We notice that MR can also be defined as $\hat{\mathbf{g}}^{l}_{lk} / \| \hat{\mathbf{g}}^{l}_{lk}\|$, where the normalization is based on the instantaneous realization of the norm, instead of the square-root of the expected norm-square. We have chosen the normalization in \eqref{eq:PrecodVecCorr} since it is the standard assumption in the Massive MIMO literature and leads to tractable analysis when dealing with spatially correlated fading channels. However, there is indications that normalization based on the instantaneous CSI can give better SEs \cite{ngo2017no}.}} 
	This precoding scheme is of the form in \eqref{eq:precoding-structure} with $\mathbf{A}_{lk} = \frac{1}{\sqrt{\mathbb{E}\{\|\hat{\mathbf{g}}^{l}_{lk} \|^{2}\}}} \mathbf{I}_M$.

	\begin{lemma} \label{Lemma:ClosedFormMRCorrelated}
		If each BS uses MR precoding, the proposed estimator can be computed as in \eqref{eq:alphalkhat} with
		\begin{equation} \label{eq:Tlk-derivation}
		\begin{aligned}
		&T_{lk} =   \sum\limits_{\substack{k'=1, \\k' \neq k}}^K \rho_{\rm dl}\eta_{lk'} \frac{\mathrm{Tr}\left(\mathbf{R}^{l}_{lk'}\boldsymbol{\Psi}^{-1}_{lk'}\mathbf{R}^{l}_{lk'}\mathbf{R}^{l}_{lk}\right)}{\mathrm{Tr}\left(\mathbf{R}^{l}_{lk'}\boldsymbol{\Psi}^{-1}_{lk'}\mathbf{R}^{l}_{lk'}\right)}+ \sum\limits_{\substack{l'=1,\\ l' \neq l}}^L \sum_{k'=1}^K \rho_{\rm dl}\eta_{l'k'} \frac{\mathrm{Tr}\left(\mathbf{R}^{l'}_{l'k'}\boldsymbol{\Psi}^{-1}_{l'k'}\mathbf{R}^{l'}_{l'k'}\mathbf{R}^{l'}_{lk}\right)}{\mathrm{Tr}\left(\mathbf{R}^{l'}_{l'k'}\boldsymbol{\Psi}^{-1}_{l'k'}\mathbf{R}^{l'}_{l'k'}\right)} \\
		&+ \sum\limits_{l' \in \mathcal{P}_{l} \setminus \{l\}} \rho_{\rm dl}\eta_{l'k} \left(\frac{\hat{p}_{lk}\tau_{p} \left| \mathrm{Tr}\left(\mathbf{R}^{l'}_{lk}\boldsymbol{\Psi}^{-1}_{lk}\mathbf{R}^{l'}_{l'k} \right)\right|^{2}  }{ \mathrm{Tr}\left(\mathbf{R}^{l'}_{l'k}\boldsymbol{\Psi}^{-1}_{lk}\mathbf{R}^{l'}_{l'k}\right)} \right)  + \sigma_{\mathrm{DL}}^2.
		\end{aligned}
		\end{equation}
	\end{lemma}
	\begin{proof}
	 {\color{black}The terms in $T_{lk}$ correspond to the inter-cell and intra-cell interference terms that commonly appear in SINR expressions. In particular, the derivation of \eqref{eq:Tlk-derivation} is similar to computation of equation $(\rm{C}.74)$ in \cite[Corollary~4.7]{bjornson2017massive} and omitted here.}
	\end{proof}
	
	The closed-form expression of $T_{lk}$ is only a function of the large-scale fading coefficients, which are constant, thus the user can compute its value once and utilize it for many coherence intervals.
	In the special case of i.i.d.~Rayleigh fading, the expression can be simplified as follows.
		\begin{corollary} \label{Lemma:ClosedFormMR}
		If each BS uses MR precoding and all channels are subject to i.i.d.~Rayleigh fading, the proposed estimator can be computed as in \eqref{eq:alphalkhat} with		
		\begin{equation}
		T_{lk} = \beta_{lk}^l \sum\limits_{k'=1,\\ k' \neq k}^K \rho_{\rm dl}\eta_{lk'} +  \sum\limits_{l'=1, l' \neq l}^L \sum\limits_{k'=1}^K \rho_{\rm dl}\eta_{l'k'} \beta_{lk}^{l'} + M\sum\limits_{l' \in \mathcal{P}_l \setminus \{l\}} \rho_{\rm dl}\eta_{l'k}  \gamma_{lk}^{l'} + \sigma_{\mathrm{DL}}^2.
		\end{equation}
	\end{corollary}

{\color{black}In this expression, the first and second terms correspond to non-coherent interference, which is similar to what has been obtained in \cite{ngo2017no}, and the third term is coherent interference coming from pilot contamination, which is introduced by orthogonal pilot signals that may be reused across cells. This expression demonstrates the impact of non-coherent interference, channel estimation quality, and pilot contamination.}
Another important precoding scheme is ZF, in which case
	\begin{equation} \label{eq:PrecodVec}
	\mathbf{w}_{lk} = 
	\frac{\mathbf{z}_{lk}} {\sqrt{\mathbb{E}\left\{\left\| \mathbf{z}_{lk} \right\|^2\right\}}} 
	\end{equation}	
	where $\mathbf{z}_{lk}$ is the $k$-th column of the pseudo-inverse matrix $\widehat{\mathbf{G}}_l \big(  \widehat{\mathbf{G}}_l^H   \widehat{\mathbf{G}}_l \big)^{-1} $ with $\widehat{\mathbf{G}}_l = [ \hat{\mathbf{g}}_{l1}^l, \ldots, \hat{\mathbf{g}}_{lK}^l] \in \mathbb{C}^{M \times K}$. 
	The expected values in \eqref{eq:TLK} cannot be computed in closed form for ZF under correlated fading, but we obtain the following result in the special case of i.i.d.~fading.	
	
	\begin{lemma} \label{Lemma:ClosedFormZF}
			If each BS uses ZF precoding and all channels are subject to i.i.d.~Rayleigh fading, the proposed estimator can be computed as in \eqref{eq:alphalkhat} with	
		\begin{equation}\label{eq:ClosedFormZF}
		\begin{aligned}
		T_{lk} &= \left(\beta_{lk}^{l'} - \gamma_{lk}^{l'} \right)\sum\limits_{\substack{k'=1,\\ k' \neq k}}^K \rho_{\rm dl}\eta_{lk'} + \sum\limits_{\substack{l' \in \mathcal{P}_{l} \setminus \{l\}}} \sum_{\substack{k'=1}}^K \rho_{\rm dl}\eta_{l'k'} \left(\beta_{lk}^{l'} - \gamma_{lk}^{l'}\right)+\sum\limits_{\substack{l' \notin \mathcal{P}_{l}}} \sum_{k'=1}^K \rho_{\rm dl}\eta_{l'k'} \beta_{lk}^{l'} \\
		& + (M-K)\sum\limits_{l' \in \mathcal{P}_{l} \setminus \{l\}} \rho_{\rm dl}\eta_{l'k} \gamma_{lk}^{l'} + \sigma_{\mathrm{DL}}^2
		\end{aligned}
		\end{equation}
		
	\end{lemma}
	\begin{proof}
		The proof follows standard derivations normally used when computing and analyzing the SINR expressions e.g. in \cite[Ch.~4]{redbook} and omitted here.		
	\end{proof}

	\section{Ergodic SE \& Asymptotic Analysis}\label{sec:SE}
	To evaluate the communication performance when using the proposed estimators of the effective downlink channel gain, we need to obtain an ergodic spectral efficiency expression that supports arbitrary estimators. The expression developed in this section will be utilized for numerical comparisons in Section~\ref{sec:numerical-results}. In this section, we will also analyze the proposed estimator's asymptotic properties as the length of the coherence intervals grows large. 
	
	\subsection{Ergodic SE}
	We begin by deriving an SE expression that can be utilized in conjunction with the proposed estimator in \eqref{eq:alphalkhat}. Since the estimate in \eqref{eq:alphalkhat} depends on the received signals, it is correlated with the data signals, which makes the derivation of the SE complicated. To resolve this issue, we follow a similar methodology as in \cite{ngo2017no}. In particular, when studying the $n$-th data symbol, we remove $y_{lk}[n]$ from the received data such that the sample average power of the signal at user~$k$ in cell~$l$ is reformulated as
	\begin{equation} \label{eq:Estv2}
	\xi^{'}_{lk}[n]  = \frac{\sum_{n'=1,n' \neq n}^{\tau_c - \tau_p} \left|y_{lk} [n']\right|^2 }{\tau_c - \tau_p-1}.
	\end{equation}
	If we utilize \eqref{eq:Estv2} to estimate the effective downlink channel gain of user~$k$ in cell~$l$, i.e., denoted as $\bar{\alpha}_{lk}^{lk}[n]$, then it is clear that  $\bar{\alpha}_{lk}^{lk}[n]$ is close to $\hat{\alpha}_{lk}^{lk}[n]$ when $\tau_c - \tau_p$ grows large. Next, we divide \eqref{eq:ReceivedSigv1} by $\sqrt{\eta_{lk}}\bar{\alpha}_{lk}^{lk}[n]$ to perform equalization of the effective channel gains; that is, making the factor in front of $s_{lk}[n]$ approximately equal to one. 
This results in the equalized received signal
	\begin{equation} \label{eq:NewReceiSig}
	\begin{split}
	y'_{lk}[n] = \frac{\alpha_{lk}^{lk}}{\bar{\alpha}_{lk}^{lk}[n]} s_{lk}[n] + \sum\limits_{\substack{k'=1,\\ k' \neq k}}^K  \sqrt{\frac{\eta_{lk'}}{\eta_{lk}}}\frac{\alpha_{lk}^{lk'}}{\bar{\alpha}_{lk}^{lk}[n]} s_{lk'}[n]
	+ \sum\limits_{\substack{l'=1,\\ l' \neq l}}^L \sum_{k'=1}^K \sqrt{\frac{\eta_{l'k'}}{\eta_{lk}}}\frac{\alpha_{lk}^{l'k'}}{\bar{\alpha}_{lk}^{lk}[n]} s_{l'k'}[n] + \frac{\tilde{w}_{lk}[n]}{\sqrt{\eta_{lk}}\bar{\alpha}_{lk}^{lk}[n]}.
	\end{split}
	\end{equation}
	The first term in \eqref{eq:NewReceiSig} contains the desired signal, while the remaining terms include mutual interference and noise. We assume that the users are capable to calculate the average of the equalized effective channel gains $\mathbb{E}\left\{\frac{\alpha_{lk}^{lk}}{\bar{\alpha}_{lk}^{lk}[n]}\right\}$, which is constant over long time and also the same for all $n$ indices, since we have i.i.d.~data symbols. Note that, the users should calculate various long-term statistics to use in the derivation herein. Hence, we can recast $y'_{lk}$ to get an equivalent expression as 
	\begin{align} \notag
	y'_{lk}[n] &= \mathbb{E}\left\{\frac{\alpha_{lk}^{lk}}{\bar{\alpha}_{lk}^{lk}[n]}\right\} s_{lk}[n] + \left(\frac{\alpha_{lk}^{lk}}{\bar{\alpha}_{lk}^{lk}[n]} - \mathbb{E}\left\{\frac{\alpha_{lk}^{lk}}{\bar{\alpha}_{lk}^{lk}[n]}\right\}\right) s_{lk}[n] + \sum\limits_{\substack{k'=1, k' \neq k}}^K  \sqrt{\frac{\eta_{lk'}}{\eta_{lk}}}\frac{\alpha_{lk}^{lk'}}{\bar{\alpha}_{lk}^{lk}[n]} s_{lk'}[n] 
	\\
	&+ \sum\limits_{\substack{l'=1, l' \neq l}}^L \sum_{k'=1}^K \sqrt{\frac{\eta_{l'k'}}{\eta_{lk}}}\frac{\alpha_{lk}^{l'k'}}{\bar{\alpha}_{lk}^{lk}[n]} s_{l'k'}[n] + \frac{\tilde{w}_{lk}[n]}{\sqrt{\eta_{lk}}\bar{\alpha}_{lk}^{lk}[n]}.
\label{eq:normalizedReceived}
	\end{align}

	In \eqref{eq:normalizedReceived}, the first term comprises the desired signal multiplied with a deterministic channel gain $\mathbb{E}\left\{\frac{\alpha_{lk}^{lk}}{\bar{\alpha}_{lk}^{lk}[n]}\right\}$. If the equalization is successful, the second term in \eqref{eq:normalizedReceived} should be small. By treating the last four terms as additive noise and applying the channel capacity bounding technique developed in \cite{jose2011pilot}, we obtain the following result.
	
\begin{lemma}\label{lemma:SE}
	A downlink ergodic spectral efficiency for user $k$ in cell $l$ is	
			\begin{equation} \label{eq:SpectralEfficiencyH}
			{\rm{R}}_{lk} = \left( 1 -\frac{\tau_p}{\tau_c} \right) \log_2 \left( 1 + \mathrm{SINR}_{lk} \right), \mbox{ [b/s/Hz]},
			\end{equation}
			where the effective downlink signal to interference and noise ratio (SINR) is given as
		\begin{equation}\label{eq:SINR}
		\mathrm{SINR}_{lk}= \frac{\left|\mathbb{E}\left\{\frac{\alpha_{lk}^{lk}}{\bar{\alpha}_{lk}^{lk}[n]}\right\}\right|^2}{\mathrm{var}\left\{\frac{\alpha_{lk}^{lk}}{\bar{\alpha}_{lk}^{lk}[n]}\right\}+ \sum\limits_{\substack{k'=1,\\ k' \neq k}}^K \frac{\eta_{lk'}}{\eta_{lk}} \mathbb{E}\left\{\left|\frac{\alpha_{lk}^{lk'}}{\bar{\alpha}_{lk}^{lk}[n]}\right|^2\right\} + \sum\limits_{\substack{l'=1,\\ l' \neq l}}^L \sum\limits_{k'=1}^K \frac{\eta_{l'k'}}{\eta_{lk}}\mathbb{E}\left\{\left|\frac{\alpha_{lk}^{l'k'}}{\bar{\alpha}_{lk}^{lk}[n]}\right|^2\right\}+ \frac{\sigma_{\mathrm{DL}}^2}{\eta_{lk}}\mathbb{E} \left\{\frac{1}{\left|\bar{\alpha}_{lk}^{lk}[n]\right|^2}\right\} }.
		\end{equation}	
	\end{lemma}
	This is an achievable SE in the sense that it is a lower bound on the ergodic channel capacity. {\color{black} In Section~\ref{sec:numerical-results}, we will demonstrate the effectiveness of this equalization-based bound compared to the conventional hardening capacity bounding technique from \cite{redbook,bjornson2017massive}. We note that \cite{ngo2017no} considers another channel capacity bound that could be applied in this setup, but it needs to be evaluated numerically using complicated numerical approximations, which makes it intractable when evaluating the SE for cellular Massive MIMO communications under the spatially correlated fading channels.}
	
For benchmark purposes, we will also consider the ideal case when the users have access to perfect CSI. Then, the first term in \eqref{eq:ReceivedSigv1} is the desired signal multiplied with a known channel and the remaining terms can be treated as additive noise. By applying a standard ergodic channel capacity bounding technique from \cite{redbook}, we have the following result:

\begin{lemma} \label{lemma:SEPerfect}
	If perfect CSI is available at the user, then the downlink ergodic spectral efficiency given as
			\begin{equation} \label{eq:SpectralEfficiencyPerfect}
			{\rm{R}}_{lk} = \left( 1 -\frac{\tau_p}{\tau_c} \right) \mathbb{E}\left\{\log_2 \left( 1 + \mathrm{SINR}_{lk} \right)\right\}, \mbox{ [b/s/Hz]},
			\end{equation}
			where the SINR is given as
			\begin{equation}\label{eq:SinrPerfect}
			\mathrm{SINR}_{lk}= \frac{\eta_{lk}\left|{\alpha_{lk}^{lk}}\right|^2}{ \sum\limits_{\substack{k'=1,\\ k' \neq k}}^K \eta_{lk'} \left|\alpha_{lk}^{lk'}\right|^2 + \sum\limits_{\substack{l'=1,\\ l' \neq l}}^L \sum\limits_{k'=1}^K \eta_{l'k'}\left|\alpha_{lk}^{l'k'}\right|^2+ \sigma_{\mathrm{DL}}^2 }.
			\end{equation}
	\end{lemma}

\subsection{Asymptotic Analysis}
Next, we will investigate the asymptotic performance of the proposed model-based estimator {\color{black}in the special case of i.i.d. Rayleigh fading channels. A similar analysis could be done for correlated Rayleigh fading but is omitted for brevity}.
Recall that we consider a setup where a set of $\tau_p$ mutually orthogonal pilots are reused among the cells. This is typically the case in cellular networks, where the number of pilots is less than the total number of users: $\tau_p < KL$.
This assumption will result in having pilot contamination in the channel estimation phase. 

We would like to analyze the term $\frac{\hat{\alpha}_{lk}^{lk}}{\alpha_{lk}^{lk}} $ that appears in the equalized received signal.
In the case of MR, we can plug $\xi_{lk}$ and $T_{lk}$, from Lemma~\ref{Lemma:asymptotic} and  Corollary~\ref{Lemma:ClosedFormMR}, respectively, to the estimator $\hat{\alpha}_{lk}^{lk}$ in \eqref{eq:alphalkhat} and rewrite it as follows
	\begin{equation}
	\begin{split}
	\frac{\hat{\alpha}_{lk}^{lk}}{\alpha_{lk}^{lk}} =&    \bigg(1 + \sum_{\substack{k'=1,\\ k' \neq k}}^K \frac{\eta_{lk'}}{\eta_{lk}}  \frac{\left|\alpha_{lk}^{lk'}\right|^2 -  \rho_{\rm dl}\beta^{l}_{lk}}{\big|\alpha_{lk}^{lk}\big|^2} +\sum_{l'=1, l'\neq l} \sum_{\substack{k'=1}}^K \frac{\eta_{l'k'}}{\eta_{lk}} \frac{\left| \alpha^{l'k'}_{lk} \right|^2 - \rho_{\rm dl}\beta^{l'}_{lk}}{\left| \alpha^{lk}_{lk} \right|^2}\\
	&+ \sum_{l' \in \mathcal{P}_{l}\setminus \{l\}} \frac{\eta_{l'k}}{\eta_{lk}}\frac{\left|\alpha_{lk}^{l'k}\right|^2 - M\rho_{\rm dl}\gamma^{l'}_{lk}}{\left|\alpha_{lk}^{lk}\right|^2} \bigg)^{1/2}.
	\end{split}
	\end{equation}
	{\color{black}
		 By replacing each $\alpha_{lk}^{l'k'}$ term with its original expression from \eqref{eq:alphaVal}, we obtain (after some straightforward algebra)
	\begin{equation}
	\begin{split}
	\frac{\hat{\alpha}_{lk}^{lk}}{\alpha_{lk}^{lk}} =&    \Bigg(1 + \sum_{\substack{k'=1,\\ k' \neq k}}^K \frac{\eta_{lk'} \gamma^{l}_{lk}}{\eta_{lk}\gamma^{l}_{lk'}}  \frac{\left|\left(\mathbf{g}_{lk}^{l}\right)^{\rm H} \hat{\mathbf{g}}_{lk'}^{l}\right|^2 -  \beta^{l}_{lk}}{\big|\left(\mathbf{g}_{lk}^{l}\right)^{\rm H} \hat{\mathbf{g}}_{lk}^{l}\big|^2} + \sum_{\substack{l'=1,\\ l'\neq l}} \sum_{\substack{k'=1}}^K \frac{\eta_{l'k'} \gamma^{l}_{lk}}{\eta_{lk}\gamma^{l'}_{l'k'}}  \frac{\left|\left(\mathbf{g}_{lk}^{l}\right)^{\rm H} \hat{\mathbf{g}}_{l'k'}^{l'}\right|^2 -  \beta^{l'}_{lk}}{\big|\left(\mathbf{g}_{lk}^{l}\right)^{\rm H} \hat{\mathbf{g}}_{lk}^{l}\big|^2}\\
	&+ \sum_{l' \in \mathcal{P}_{l}\setminus \{l\}} \frac{\eta_{l'k} \gamma^{l}_{lk}}{\eta_{lk} \gamma^{l'}_{l'k}}\frac{\left|\left(\mathbf{g}_{lk}^{l'}\right)^{\rm H} \hat{\mathbf{g}}_{l'k}^{l'}\right|^2 -M  \gamma^{l'}_{lk}}{\big|\left(\mathbf{g}_{lk}^{l}\right)^{\rm H} \hat{\mathbf{g}}_{lk}^{l}\big|^2} \Bigg)^{1/2}.
	\end{split}
	\end{equation}}

	 {\color{black} 
	 	For the ratios in the first two summations, we have 
	 	\begin{equation}
	 	 \frac{\eta_{lk'} \gamma^{l}_{lk}}{\eta_{lk}\gamma^{l}_{lk'}}\left(  \frac{\left|\frac{\left(\mathbf{g}_{lk}^{l}\right)^{\rm H} \hat{\mathbf{g}}_{lk'}^{l}}{M}\right|^2 -  \frac{\beta^{l}_{lk}}{M^2}}{\big|\frac{\left(\mathbf{g}_{lk}^{l}\right)^{\rm H} \hat{\mathbf{g}}_{lk}^{l}}{M}\big|^2}\right) \rightarrow 0
	 	\end{equation}
	 	  
	 	 \begin{equation}
	 	 \frac{\eta_{l'k'} \gamma^{l}_{lk}}{\eta_{lk}\gamma^{l'}_{l'k'}}\left(  \frac{\left|\frac{\left(\mathbf{g}_{lk}^{l}\right)^{\rm H} \hat{\mathbf{g}}_{l'k'}^{l'}}{M}\right|^2 -  \frac{\beta^{l}_{lk}}{M^2}}{\big|\frac{\left(\mathbf{g}_{lk}^{l}\right)^{\rm H} \hat{\mathbf{g}}_{lk}^{l}}{M}\big|^2}\right) \rightarrow 0
	 	  \end{equation}
		  as $M \rightarrow \infty$, because the terms in the numerators converge to zero, while the terms in the denominator have finite limits. However, for the last summation, as $M \rightarrow \infty$, we have 	
	 	\begin{equation}
	 	\frac{\eta_{l'k} \gamma^{l}_{lk}}{\eta_{lk}\gamma^{l'}_{l'k'}}\left(  \frac{\left|\frac{\left(\mathbf{g}_{lk}^{l'}\right)^{\rm H} \hat{\mathbf{g}}^{l'}_{l'k}}{M}\right|^2 -  \frac{M\gamma^{l'}_{lk}}{M^2}}{\big|\frac{\left(\mathbf{g}_{lk}^{l}\right)^{\rm H} \hat{\mathbf{g}}_{lk}^{l}}{M}\big|^2}\right) \rightarrow \frac{\sqrt{\hat{p}_{lk}}\eta_{l'k} \beta^{l'}_{lk}}{\sqrt{\hat{p}_{l'k}}\eta_{lk}\beta^{l'}_{l'k}},
	 	\end{equation}
	 	where we used the relation  \eqref{eq:ChannelEstRelation} to identify the fact that pilot contamination makes the first term in the numerator non-zero. This indicates that the ratio between the estimates and the  actual effective channel gain will not converge to one, due to pilot contamination. For example, there is always a risk that $\xi_{lk}  < T_{lk}$, so that the proposed estimator cannot extract any useful information from $\xi_{lk} $.}
	
	To demonstrate this effect when having a finite number of BS antennas, Fig.~\ref{fig:bound} displays the cumulative distribution function (CDF) of $\xi_{lk}$ for different realizations of small-scale fading, and the constant value of $T_{lk}$ is shown as a vertical line. A single setup with fixed large-scale fading coefficients, $M=100$, MR precoding is considered and $3$ users per cell. It can be seen that  $\xi_{lk}$ is not always greater than $T_{lk}$ due to the pilot contamination effect, which makes the fluctuations in the interference from the signals to pilot-sharing users relatively large. {\color{black}
		Even asymptotically as $M \to \infty$, the user is unable to separate this interference from the desired signal and there is a non-zero probability that $\xi_{lk}  < T_{lk}$ in a given coherence interval.
	The same thing would happen if ZF precoding was used.
	In the proposed method, we address this issue by selecting the switching point $\Theta_{lk}$ in \eqref{eq:alphalkhat} such that $ \xi_{lk}  - T_{lk} > 0$ so that the square-root provides a real number.
	The SE maximizing way of selecting the switching point is to evaluate the SE expression in Lemma~\ref{lemma:SE} for both cases in \eqref{eq:alphalkhat} and select the one providing the largest value. This approach will be used later in Section \ref{sec:numerical-results} when evaluating the SE performance of the model-based approach.} In Section \ref{sec:datadriven}, we design an alternative blind estimator by using neural networks.

	\begin{figure}[htb!]
		\centering
		\includegraphics[width=.6\columnwidth]{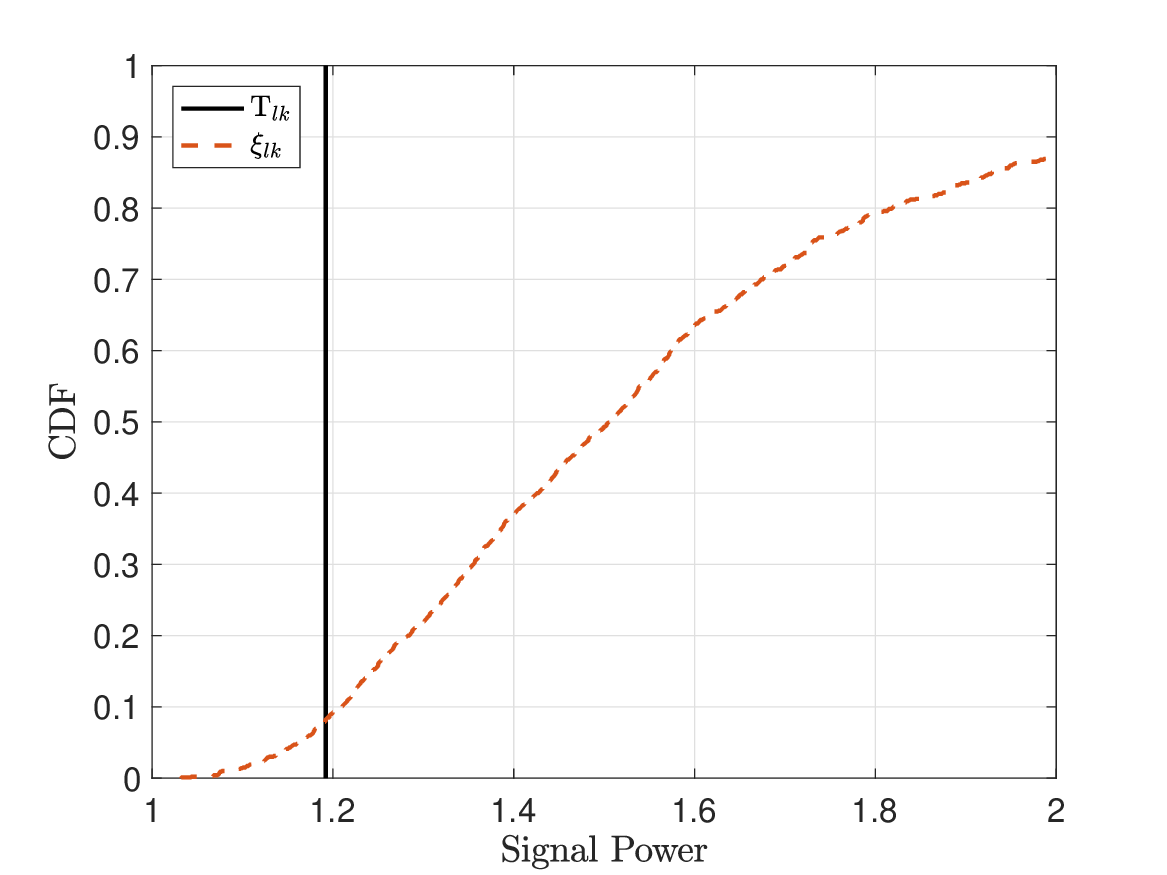}
		\caption{Comparing calculated power with $\xi_{lk}$ vs $T^{lk}_{lk}$.}
		\label{fig:bound}
	\end{figure}
	
\begin{remark}
	{\color{black} The model based approach is derived based on asymptotic properties that appear when $\tau_c,K \rightarrow \infty$, which is a common approach for developing blind algorithms.
	This implies that there is no guarantee that the estimator performs well when the number of users is small or the coherence interval is short. However, aligned with the previous works that apply random matrix theory to analyze Massive MIMO systems (e.g.,  \cite{6415388}), we expect the framework to work quite well when having a small number of users and practical coherence interval. To investigate how much better estimation performance can be achieved in practice, in Section \ref{sec:datadriven}, we design an alternative blind estimator by using neural networks and train them for the  non-asymptotic case.}
\end{remark}

\section{Data-driven Approach for Downlink Massive MIMO Channel Estimation}\label{sec:datadriven}

  {\color{black} The proposed blind downlink channel estimator in \eqref{eq:alphalkhat} is model-aided, in the sense that it was developed by studying the asymptotic properties of the system model. While the estimator is expected to work well when the coherence interval is large and there are many users per cell, there is no guarantee that the estimator will work well under the circumstances that occur in practical Massive MIMO systems. For example, the number of users per cell might be small, in particular, under low-traffic hours or when the coherence interval is relatively small. In this section, we take an alternative data-driven approach where we identify a mapping between the available information at the UE and the downlink effective channel gain. 
  More precisely, we tackle the mentioned limitations of the proposed model-based blind downlink channel estimator by training a fully-connected neural network for the same task.
  The goal is to determine under what conditions and to what extent the proposed model-aided estimator can be outperformed.}

{\color{black}  The universal approximation theorem states that one can approximate any continuous function between a given input vector and a desired output vector arbitrarily well using a sufficiently large fully-connected neural network \cite{goodfellow2016deep,hornik1989multilayer}. However, this theorem does not provide any exact details on the neural network structure (e.g., the number of layers and neurons) or what algorithms to utilize to find the optimal approximation. This effort must be carried out for every problem at hand.\footnote{\color{black} One can make an analogy to function approximation using polynomials: it is known that a continuous function can be approximated arbitrarily well by a polynomial, but we still need to find the right order and coefficients.}  Notice that, there may exist more than one neural network for a given prediction accuracy, but 
	in this paper, we use a fully-connected feed-forward neural network to estimate $\alpha^{lk}_{lk}$ from input data available at an arbitrary user $k$ in cell $l$ due to its simplicity.} As input to the neural network, we consider three features, $\xi^{'}_{lk}[n]$ that is given in \eqref{eq:Estv2}, $T_{lk}$ provided in \eqref{eq:TLK}, and $\eta_{lk}\rho_{\rm dl}\beta^{l}_{lk}$ for the user $k$ in cell $l$. The input is selected to enable the network to learn about the pathloss model, propagation environment, and mapping between the sample average power and effective channel gains. The input vector to the neural network is denoted as $\boldsymbol{\kappa}_{lk} \in \mathbb{R}^{P}$, where $P=3$ in the proposed design. The output is a scalar  $o_{lk}$ that is supposed to be equal to the absolute value  of downlink effective channel gain $\alpha^{lk}_{lk}$. 

\begin{remark}
	{\color{black} The reason for selecting the three features, $\xi^{'}_{lk}[n]$, $T_{lk}$, and $\eta_{lk}\rho_{dl}\beta_{lk}^{l}$, in the data-driven approach is to enable the neural network to learn the pathloss and propagation environment through  $\eta_{lk}\rho_{dl} \beta_{lk}^{l}$ and to also identify a mapping between the sample average power and the effective channel gains by using the received  average power $\xi^{'}_{lk}[n]$ and $T_{lk}$. 
	We have purposely kept the number of input features small to make it feasible to implement the neural network on user devices. 
	There are other relevant features that can be utilized, but we selected the set of features that provided the best performance during the feature selection process.}
\end{remark}

We consider a fully-connected neural network with $I$ layers to approximate the ideal  non-linear mapping from $\boldsymbol{\kappa}_{lk}$ to $o_{lk}$ \cite{o2017introduction,demir2019channel}. 
We designed our network with $I = 3$ hidden layers and the size of each layer is specified in Table \ref{table:1}.
{\color{black}The number of layers and neurons per layer were selected by taking a trial and error approach. In particular, we selected the number of neurons per layer and the number of layers by increasing them slowly and choosing different activation functions to find a network structure which offers good performance in terms on NMSE. The layout provided in Table \ref{table:1} with Rectified Linear Unit (ReLU) activation is the successful candidate. The details of the other parameters settings for the neural network are provided in Section \ref{sec:numerical-results}.}

{\color{black}We trained the neural network by taking a supervised deep learning approach.\footnote{\textcolor{black}{Supervised deep learning is the best option when exact downlink effective gains are available in the training stage. More efforts are required to achieve good prediction performance in case of imperfect effective gains, for example, using regularization or a hybrid between supervised and unsupervised deep learning approaches.}} To this end, we created a labeled training set consisting of inputs and corresponding optimal outputs (i.e., true effective channel gain) pairs. The training set is defined as $\{\boldsymbol{\kappa}^{d}_{lk}, \hat{o}^{d}_{lk}\}^{D}_{d=1}$, where  $D$ represents the number of points in the set. For training example $d$, $ \boldsymbol{\kappa}^{d}_{lk}$ is the input vector and the corresponding desired output is $ \hat{o}^{d}_{lk}$ \cite{o2017introduction,demir2019channel}. For the loss function, we exploit the mean absolute error (MAE) for the training stage as
	\begin{equation}
	\mathcal{L} (\mathsf{\Theta}) = \mathbb{E} \{ |\tilde{\alpha}_{lk}^{lk} - \alpha_{lk}^{lk} | \},
	\end{equation}
	where $\mathsf{\Theta}$ is the set of weights and biases assigned to the fully-connected neural networks and $\tilde{\alpha}_{lk}^{lk}$ is the effective downlink channel gain obtained by the neural network. Despite a high cost for multi-cell Massive MIMO communications, the perfect gain $\alpha_{lk}^{lk}$ can be obtained if all the users use orthogonal pilot signals with sufficiently large pilot power \cite{ngo2017no}. We refer to textbooks such as \cite{goodfellow2016deep} for a detailed description of the training process for neural networks.}

{\color{black}We generated the training data set based on the system model presented in the paper, and detailed design settings are provided in Section \ref{sec:numerical-results}. {\color{black} Since user locations and shadow fading coefficients are identically and independently distributed, the dataset consists of input-output pairs for a typical user $k$, which refers to a user, randomly located in cell $l$ for $2000$ large-scale and $1000$ small-scale fading realizations corresponding to $D = 2000000$ samples. Note that since we train a network for a \emph{typical} user, the same trained network can be utilized for every user in the cellular network. {\color{black} All simulations for generating the data set are implemented in MATLAB on a MacBook Pro laptop with CPU Intel Core i5, $2.3$ GHz Quad-Core, and $16$ GB RAM. The generating time for the data set with $2$ million samples depends on the number of users per cell. In the case of $K=3$, each sample takes $3.3$ milliseconds per sample which corresponds to less than two hours of simulation time, and for $K = 10$, it takes $12.5$ milliseconds per sample which approximately takes $7$ hours simulation time.}} The training is done offline, but the actual usage of the neural network is when it is used by the users in a cellular system, in online mode. In addition, in this work, the data is generated from the simulation setup. However, the input data is selected such that it is practical to obtain such data from measurements that can be made in a practical setup. The main challenge is to obtain the labels, but one potential solution is that the cellular system occasionally transmits orthogonal pilot sequences of length $\tau_{c}-\tau_{p}$, in an entire coherence interval in the downlink. These pilot sequences can be reused sparsely in the network (e.g., reuse 7) so that there is essentially no pilot contamination, and the SNR will be very high after despreading, so that the true $\alpha^{lk}_{lk}$ can be estimated accurately. These sequences can also be utilized to estimate and calibrate other aspects of the system.}

\begin{table}[t]
	\centering
	\caption{Layout of the deep neural network.} 
	\begin{center}
		\begin{tabular}{| c| c| c| c|}
			\hline
			& Neurons & Parameters& Activation function\\ 
			\hline
			Layer 1 & 32 &  256 & ReLU \\
			\hline
			Layer 2 & 64 & 2112 & ReLU \\
			\hline
			Layer 3 & 64 & 4160 & ReLU \\
			\hline
		\end{tabular}\label{table:1}
	\end{center} 
\end{table} 

 {\color{black}To evaluate the SE performance of the data-driven approach, we can derive a similar downlink ergodic SE expression as given in Lemma \ref{lemma:SE}. The derivations follow similar lines as the model-aided estimator. Here, we perform received signal equalization in \eqref{eq:NewReceiSig} with the effective gain estimates from the data-driven approach, and the rest of the derivations follows as the model-based approach. Note that the trained neural network for $n$th data symbol is applicable for all as we have i.i.d. symbols. Hence, there is no need to train neural network $n$ times in each coherence interval, and the bound in Lemma \ref{lemma:SE} provides a lower bound on the ergodic channel capacity.}
	
\begin{remark}
	{\color{black} This paper focuses on the downlink channel estimation quality in a comparison between the model-based and data-driven approaches for given transmit power coefficients.  However, an extension to the power allocation for a specific utility metric should be interesting for future work. The main structure of the fully-connected neural network might be kept the same, but a fine-tuning should be made to learn new features from the optimal power allocation.}
\end{remark}

\section{Numerical Results} \label{sec:numerical-results}
	This section provides a numerical comparison of our proposed model-based and data-driven approaches, as well as a comparison with the conventional hardening bound that uses $\mathbb{E}\{ \alpha_{lk}^{lk} \}$ as the estimate of $\alpha_{lk}^{lk}$ when the user decodes the downlink data. In the simulation setup, we consider a multi-cell Massive MIMO network consisting of $4$ cells. We model the cellular network with a grid layout in a 500\,m $\times$ 500\,m area where each square cell has a BS in the center. We use the wrap-around technique to avoid edge effects. Each BS is serving $K$ users, which are uniformly distributed in the coverage area of their serving BS while keeping a minimum distance of $35\,$m. Furthermore, each BS has $M = 64$ antennas, and each coherence interval contains $500$ symbols {\color{black} and we have pilot reuse factor $f=1$}. We model the macroscopic large-scale fading coefficients as \cite{bjornson2017massive}
	\begin{equation}\label{eq:largeScale}
	\beta^{l}_{l'k'} \left[{\rm dB}\right] = -35 - 36.7\log_{10}\left(d^{l}_{l'k}/1\,\rm{m}\right) + F^{l}_{l'k'} ,
	\end{equation}
	$d^{l}_{l'k'}$ denotes the distance between user $k'$ located in cell $l'$ to BS $l$ and $F^{l}_{l'k'}$ is shadow fading generated from a log-normal distribution with standard deviation of $7\,$dB and it is generated independently for each user. 
	To ensure that each user has its largest large-scale fading from the BS in its own cell, we regenerated the shadow fading realizations whenever this was not the case. It means that $\beta^{l}_{lk}$ is the largest among all $\beta^{l'}_{lk}$, $l'  = 1,\dots,L$. We consider communication over a $20\,$MHz bandwidth and the noise variance is $-94\,$dBm. We assume equal power control scheme in the downlink data transmission, and the uplink transmit power of users is set to $100\,$mW. We assume that each BS equipped with a horizontal uniform linear array with half-wavelength antenna spacing and the spatial correlation matrix of user $k$ located in cell $l'$ to the BS $l$ is modeled by the approximate Gaussian local scattering model provided in \cite[Ch.~2.6]{bjornson2017massive} with $(m,n)$th elements given by
	\begin{equation} \label{eq:LSM}
	\left[\mathbf{R}^{l}_{l'k}\right]_{m,n}  =\beta^{l}_{l'k} e^{\pi j (m-n)\sin (\varphi^{l}_{l'k})} e^{- \frac{\sigma^2_{\varphi}}{2} (\pi (m-n)\cos (\varphi^{l}_{l'k}))^2}.
	\end{equation}
	This model assumes that each user has a scattering cluster around it while there are no other scattering clusters. In \eqref{eq:LSM}, $\varphi^{l}_{l'k}$ denotes the nominal angle of arrival (AoA) from user $k$ in cell $l'$ to the BS $l$. It is also assumed that the multipath components of a cluster have Gaussian distributed AoA around nominal AoA with an angular standard deviation (ASD) $\sigma_{\varphi} = 7$ degree in the simulations. Note that for the uncorrelated fading channel, we generate the spatial correlation matrices as 
	\begin{equation}
	\mathbf{R}^{l}_{l'k}  =\beta^{l}_{l'k} \mathbf{I}_{M}
	\end{equation}
	where the large-scale fading coefficients i.e., $\beta^{l}_{l'k}$ are modeled as in \eqref{eq:largeScale}. The reason to consider both cases of uncorrelated and correlated fading models is to highlight the differences between these models. In particular, the uncorrelated fading case offers the highest channel hardening level, while the correlated fading case might feature little hardening when the ASD is small.

To generate the data-driven simulation results, we used a fully-connected neural networks model. The detailed design specifications of the layout including number of neurons per layer, each layer's activation functions and number of parameters are provided in Table \ref{table:1}. {\color{black} The entire data set consists of $2000000$ input-output vector pairs from which, we selected $400000$ for training, $100000$ for validation and the rest $1500000$ for the testing phase. {\color{black} The neural network model is implemented by using The Keras open-source library in Python and the training time complexity for $400000$ training and $100000$ validation samples is around $150$ seconds.}} In the neural network model's training phase, we select Adam optimizer \cite{kingma2017adam}, and the loss function is set to MAE. The hyper-parameters in the model are selected as the following: learning rate of Adam optimizer's is $0.01$, and batch size is $128$, and the number of epochs is equal to $200$.\footnote{{\color{black} In this paper, we assume that the training phase is implemented offline so the network can afford a fixed training rate. In our simulations, we have observed that a learning rate of $0.01$ gives good results. However, an adaptive learning rate may accelerate the training phase.}}

To evaluate the performance of the proposed model-based and data-driven approaches, we investigate two important metrics: the NMSE and SE. The NMSE of a user $k$ in cell $l$ is defined in \eqref{eq:NMSE}. We plot the CDF of NMSE when the median downlink SNR, i.e., $\rm{SNR_{dl}}$ of a cell-edge user is $10\,$dB. The two different benchmarks  \textcolor{black}{``Hardening bound", ``$\tau_c = \infty$"} and the results of two proposed approaches  \textcolor{black}{``Proposed: model-aided", ``Proposed: data-driven"}, are included for comparison defined as
\begin{enumerate}
	\item The state-of-the-art named as \textcolor{black}{``Hardening bound"}, that uses $\mathbb{E}\{ \alpha_{lk}^{lk} \}$ as the estimate of $\alpha_{lk}^{lk}$, which has been popularly used in the Massive MIMO literature, for example \cite{sanguinetti2019toward, yang2017massive} and references therein.
	\item In  $\tau_c = \infty$, it is assumed that the user knows the asymptotic value of $\xi_{lk}$ thanks to an infinite time interval, which is used to study the asymptotic behavior of the proposed model-based approach  \cite{ngo2017no}.
	\item The \textcolor{black}{``Proposed: model-aided"} and \textcolor{black}{``Proposed: data-driven"} are the results from our model-based and data-driven approaches, respectively. 
\end{enumerate}

Figs.~\ref{fig:correlatedMRMSE} and \ref{fig:correlatedZFMSE} depict the NMSE at the BSs for correlated Rayleigh fading channel models with MR and ZF precoding, respectively. The curve for the hardening bound is the rightmost for both cases which shows the worst NMSE performance. 
It can be seen from the figures that ``Proposed: model-aided'' and ``Proposed: data-driven'' perform better than the hardening bound for both MR and ZF. In the case of MR, data-driven approach is performing better than model-aided and asymptotic results. In ZF, the result for both proposed methods are comparable to each other and they perform close to the curve corresponding to asymptotic result.

\begin{figure*}[t]
	\begin{minipage}{0.48\textwidth}
		\centering
		\includegraphics[trim=0.5cm 0.0cm 1.5cm 0.8cm, clip=true, width=3.0in]{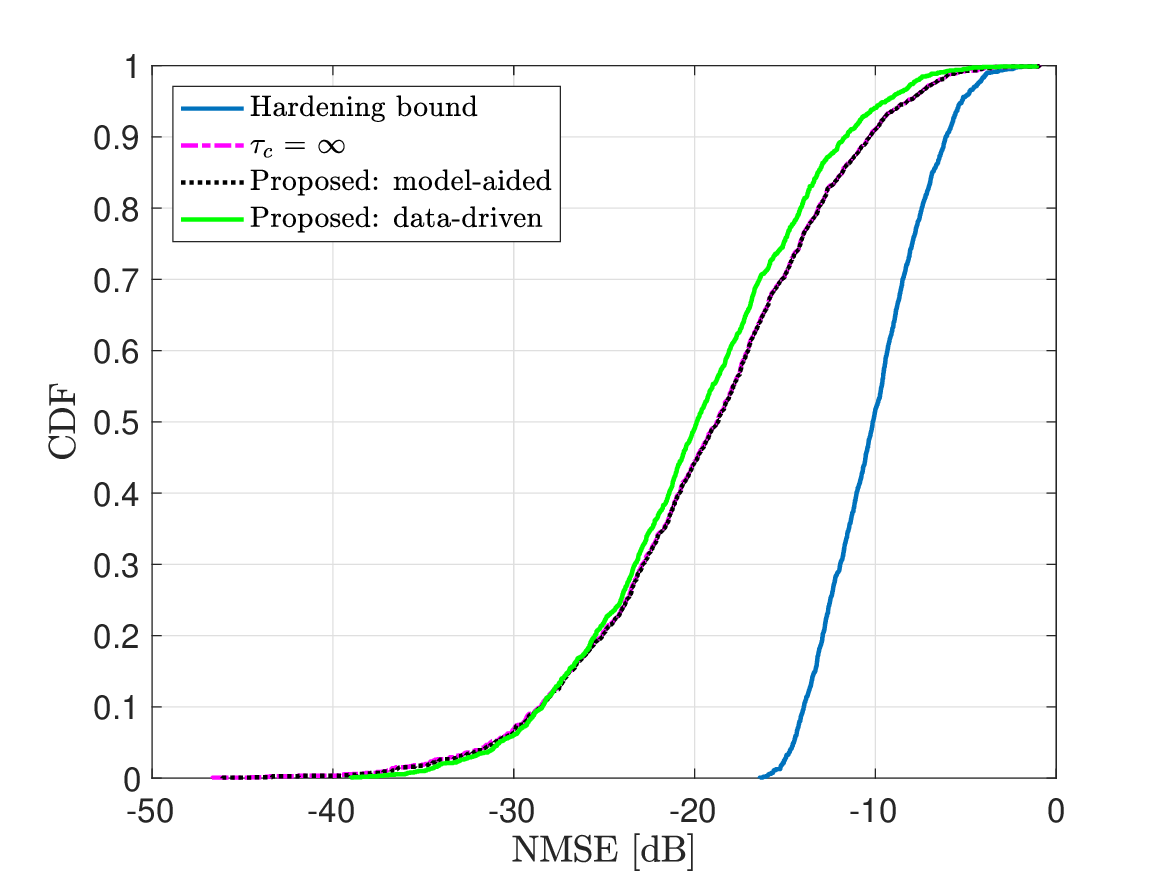} \vspace{-0.5cm}
		\caption{The NMSE for correlated channel model and MR precoding at BSs.}
		\label{fig:correlatedMRMSE}
	\end{minipage}
	\hfill
	\begin{minipage}{0.48\textwidth}
		\centering
		\includegraphics[trim=0.5cm 0.0cm 1.5cm 0.8cm, clip=true, width=3.0in]{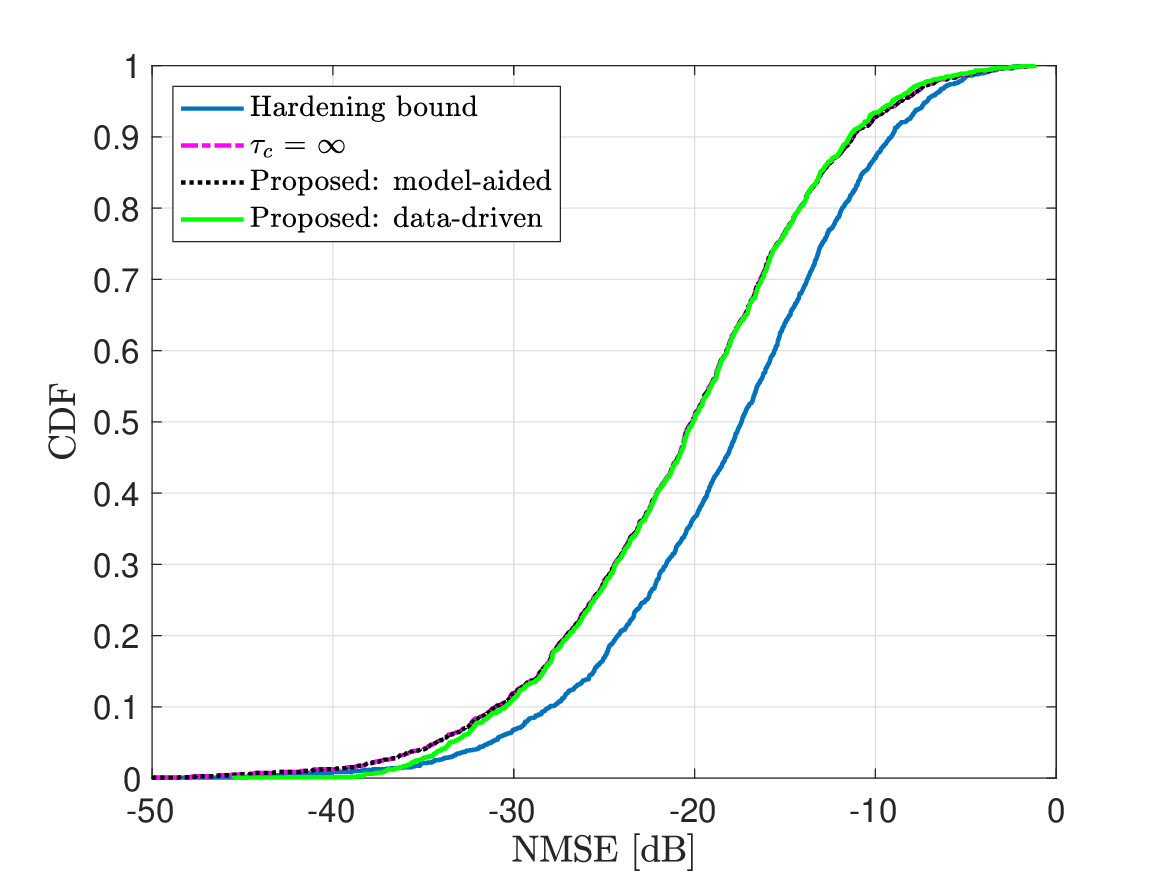} \vspace{-0.5cm}
		\caption{The NMSE for correlated channel model and ZF precoding at BSs.}
		\label{fig:correlatedZFMSE}
	\end{minipage}
\end{figure*}

The NMSE results are mainly targeted to provide a rough comparison of the proposed model-aided and data-driven approaches with the conventional hardening approach that uses $\mathbb{E}\{ \alpha_{lk}^{lk} \}$ as the estimate of $\alpha_{lk}^{lk}$. To determine the performance of the proposed approaches in practical systems, we need to compare the SE that they are delivering. 
Figs.~\ref{fig:SE_MR_K310} and \ref{fig:SE_ZF_K310} show the CDF of the SE per user for MR and ZF precoding at the BSs, respectively, when the median DL SNR is $10$\,dB for cell-edge users. 
For MR precoding, Fig.~\ref{fig:SE_MR_K310} shows that there is a significant SE improvement when using the model-based approach compared to hardening bound. The performance improvement is particular large for users with good channel conditions.
This shows that the conventional hardening bound can greatly underestimate the achievable performance over Massive MIMO channels that feature little channel hardening, which is the case for our considered spatially correlated fading channel model with a small ASD. 
This assumption results in having a low-rank correlation matrix with a few dominant eigenvalues.
The data-driven approach results in better SE than the model-based approach in the lower 40 \% of the CDF curve and comparable for the other 60 \%. As a reference, we also show the SE that is achievable by utilizing the bound provided in Lemma  \ref{lemma:SEPerfect}, when knowing the channel gains perfectly at the user (denoted as perfect CSI) and there is a significant performance difference.

\begin{figure*}[t]
	\begin{minipage}{0.48\textwidth}
		\centering
		\includegraphics[width=1.0\columnwidth]{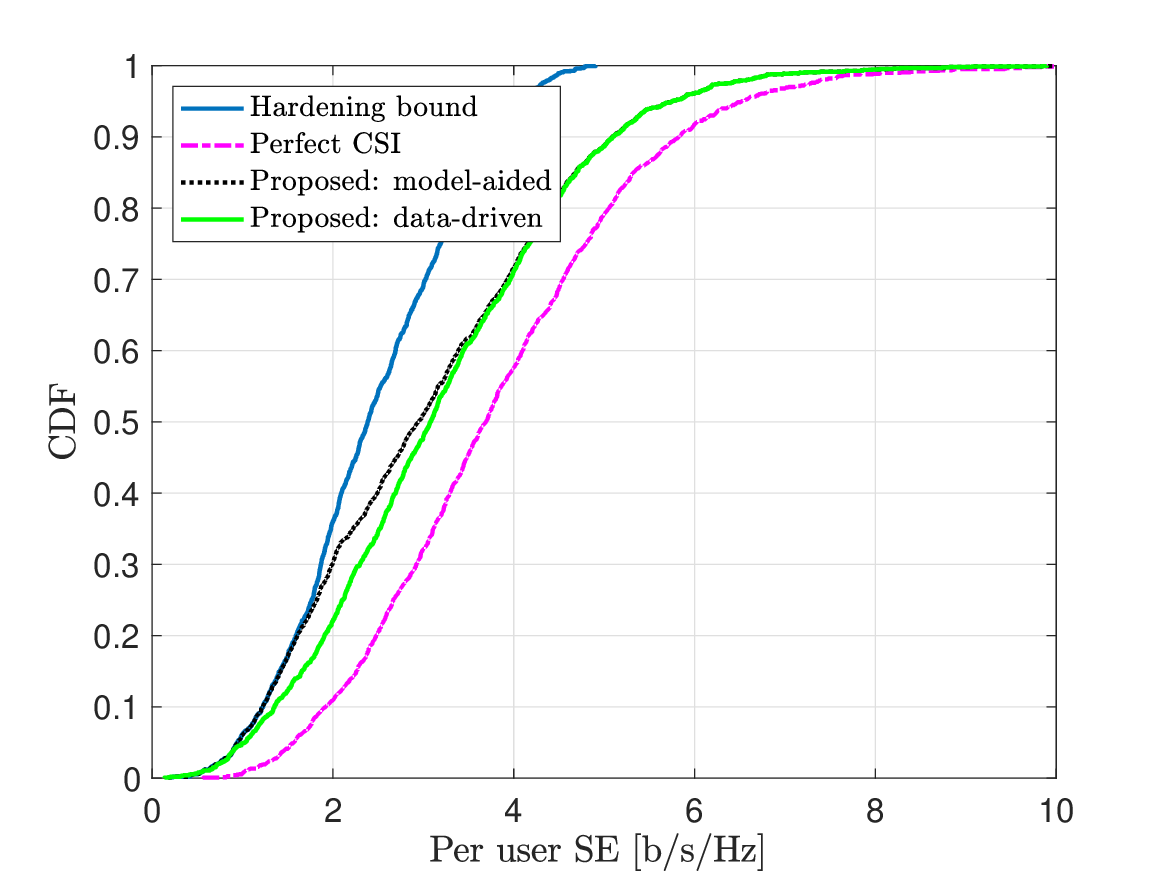}
		\caption{CDF of the SE per user for the correlated channel model and MR precoding, with DL SNR = $10\,$dB and $K = 3$.}
		\label{fig:SE_MR_K310}
	\end{minipage}
	\hfill
	\begin{minipage}{0.48\textwidth}
		\centering
		\includegraphics[width=1.0\columnwidth]{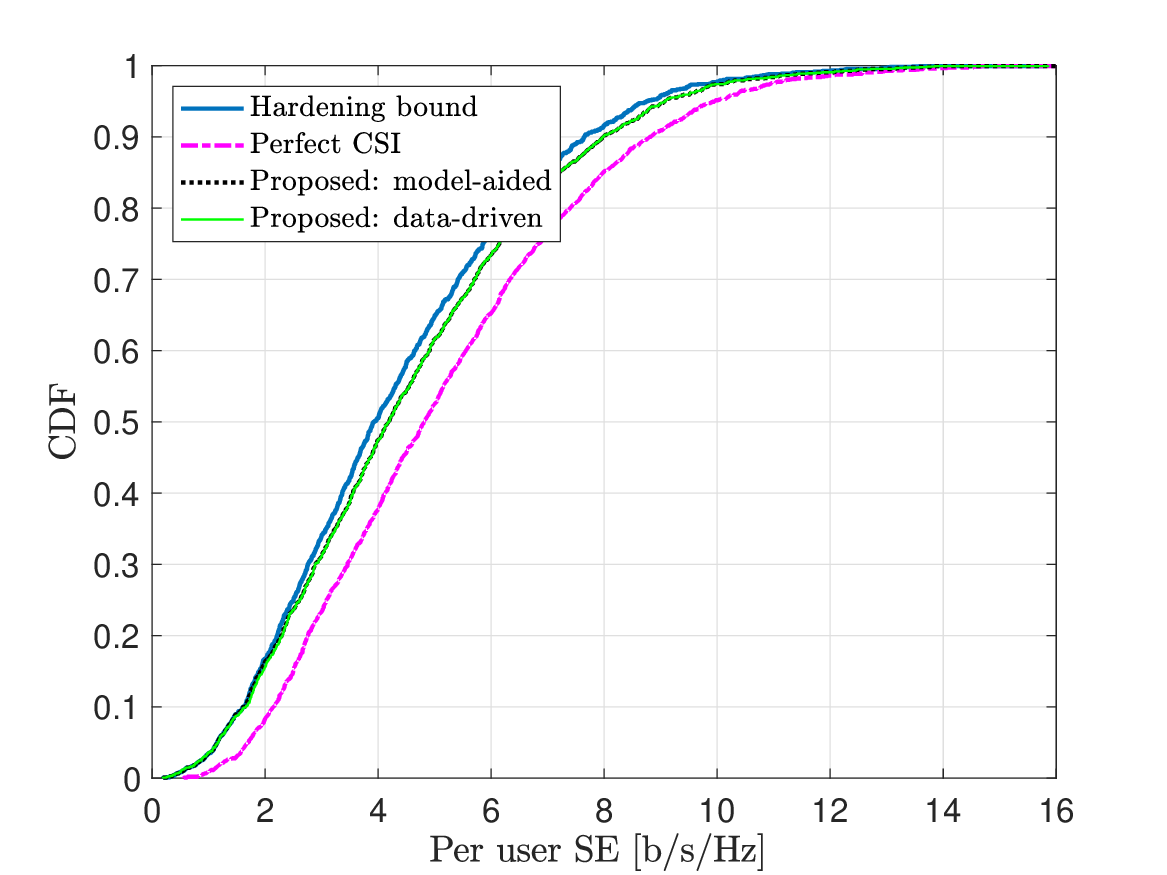}
		\caption{CDF of the SE per user for the correlated channel model and ZF precoding, with DL SNR = $10\,$dB and $K = 3$.}
		\label{fig:SE_ZF_K310}
	\end{minipage}
\end{figure*}

The ZF precoding results in Fig.~\ref{fig:SE_ZF_K310} show that the performance gap between the proposed hardening bound and the model-based and data-aided estimator is smaller. Note that the limitation of using the mean of effective channel gains for decoding the desired data in low-hardening channels is more of a physical limitation. However, our proposed approaches are applicable for low- and high-hardening models and are more useful for the low-hardening channel conditions. The selection of precoding vectors and normalization of them is the network design preference which can affect the performance behavior. In this paper we utilized average-normalize approach to normalize the precoding vectors as give in \eqref{eq:PrecodVecCorr} and \eqref{eq:PrecodVec}. However, one can design the precoding vectors with different normalization.

Furthermore, MR and ZF precoding are behaving differently on favoring users with strong and weak channels. Hence, one should consider this for having a fair comparison of their performance. In addition, different normalization of precoding vectors and SE bounding techniques can potentially show different performances.  Using ZF with the average-normalized tends to give a closer performance for  our proposed approach and hardening bound. Therefore, for the average-normalized approach used in this paper, all the considered benchmarks are close to the perfect CSI with ZF precoding. Therefore, there is less room for them to show more distinction. Normalization plays an important role in the behavior seen in this results and, using different vector normalization for precoding vectors can potentially show different performance behavior. Hence, the precoding technique and normalization should be considered as the design preference of the networks.

{\color{black}Next, in order to analyze the performance of the proposed model-aided and data-driven channel estimation methods in more diverse scenarios, we investigate the effect of increasing the number of users per cell from $K = 3$ to $K = 10$. The results are provided in  Figs. 6 and 7 for MR and ZF precoding, respectively. In addition, as we assumed pilot reuse factor $f=1$ in the simulation setup and $\tau_{p} = fK$. Increasing K reflects that we increase $\tau_{p}$ as well. Therefore the result reflects the effect of having different length of pilots on the performance of proposed methods. It can be seen that,} the gap between the proposed approaches and the case with perfect channel knowledge has shrinked. The estimated effective channel gain is getting closer to its asymptotic limit by increasing the number of users, resulting in a comparable performance for the ergodic SE given in Lemma \ref{lemma:SE} and hardening bound. In addition, comparing with the case of $3$ users, the SEs are decreasing, which indicates that interference is becoming more dominant which is also affecting the result of Lemma \ref{lemma:SEPerfect}, denoted as perfect CSI. Note that the results of data-driven for $K=10$ are obtained by utilizing the trained model for the case of $K=3$. This indicate that data-driven model is robust towards change in the number of users in the cellular network.

\begin{figure*}[t]
	\begin{minipage}{0.48\textwidth}
	\centering
	\includegraphics[width=1\columnwidth]{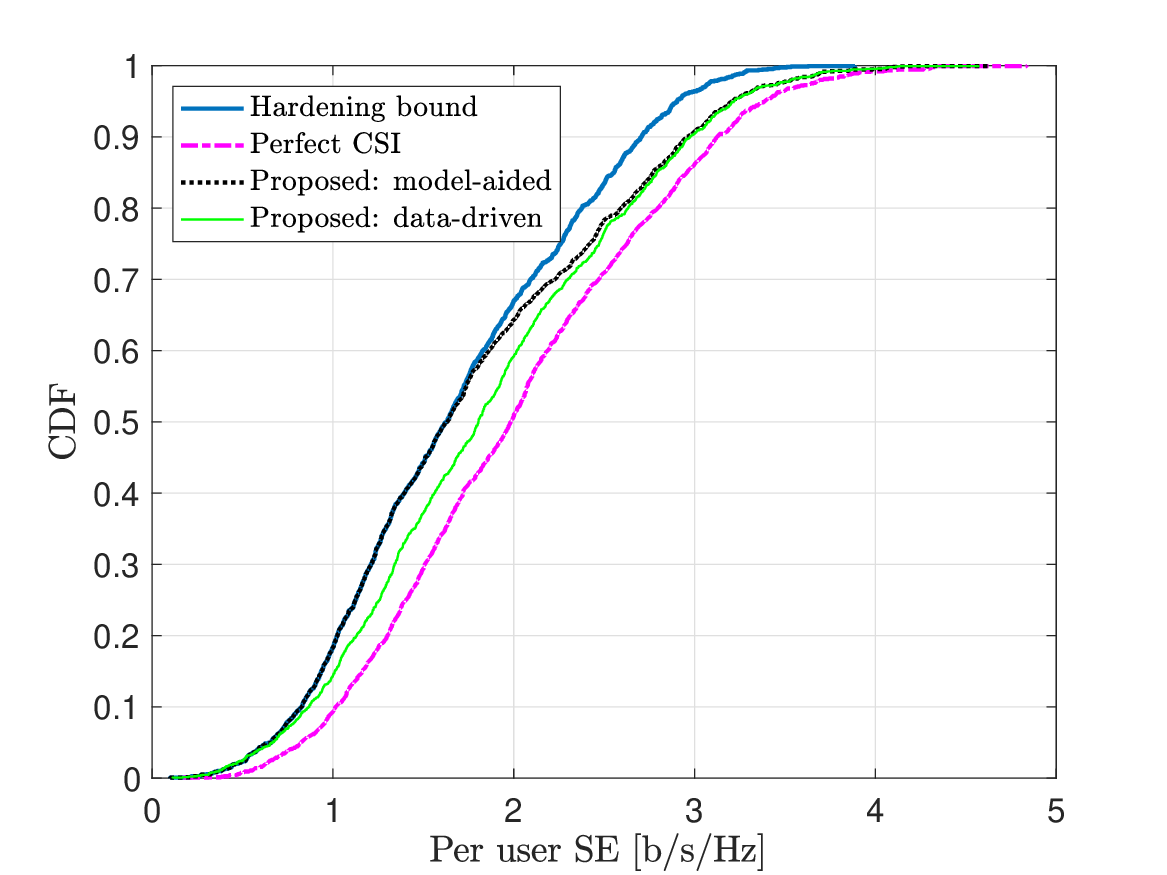}
	\caption{CDF of the SE per user for the correlated channel model and MR precoding, with DL SNR = $10\,$dB and $K = 10$.}
	\label{fig:SE_MR_K1010}
	\end{minipage}
	\hfill
	\begin{minipage}{0.48\textwidth}
	\centering
	\includegraphics[width=1\columnwidth]{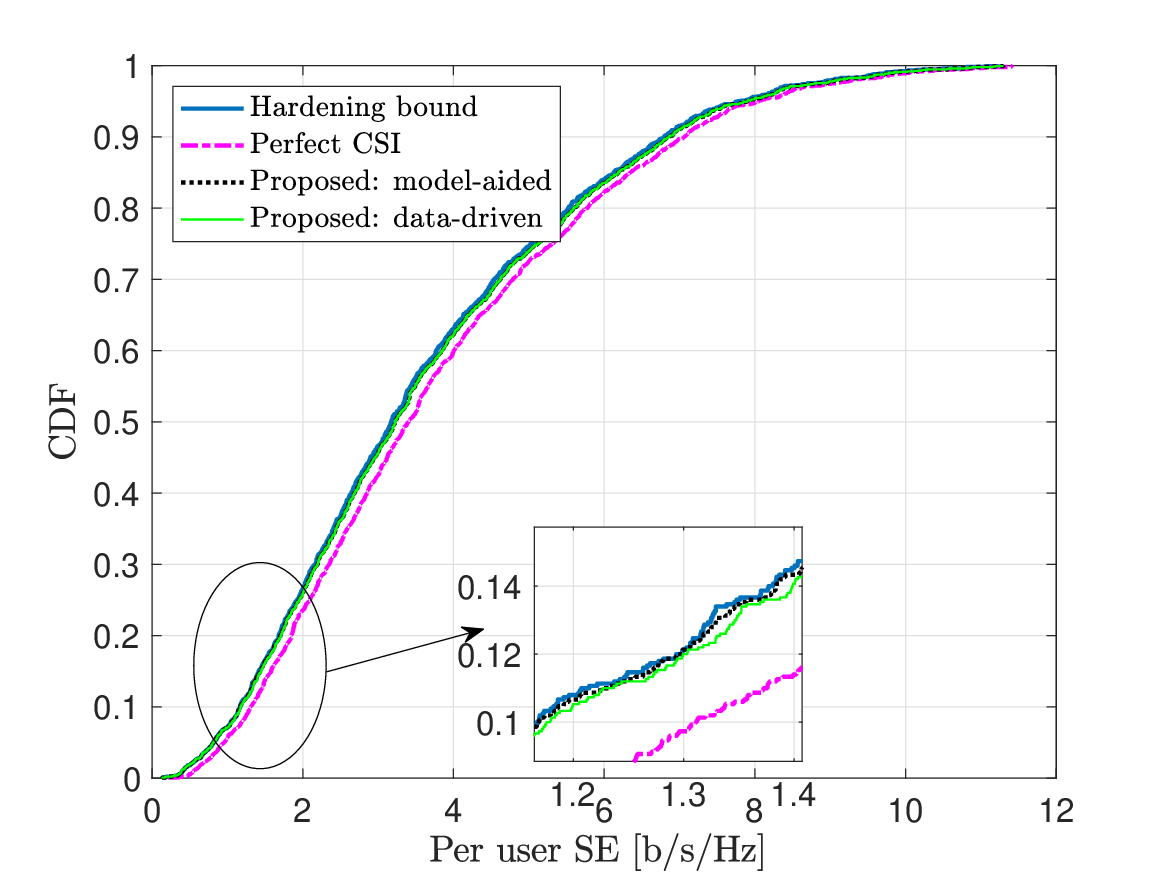}
	\caption{CDF of the SE per user for the correlated channel model and ZF precoding, with DL SNR = $10\,$dB and and $K = 10$.}
	\label{fig:SE_ZF_K1010}
	\end{minipage}
\end{figure*}
Using the mean of effective channel gain to decode the desired data signal at the user is a legitimate assumption when the channel hardening holds. To show this, we provide per user SE results for the case of uncorrelated Rayleigh fading channel model in Figs.~\ref{fig:SE_un_MR_K30} and \ref{fig:SE_un_ZF_K30}. \textcolor{black}{There is a notable gap between the proposed approaches and the hardening bound when considering MR, while the gap is rather small when considering ZF.
The reason is that MR creates large variations in the effective channel gain by assigning more power when the small-scale fading gives a strong realization and less power when the realization is weak. ZF does the opposite and, therefore, provides a better channel hardening behavior.}

\begin{figure*}[t]
	\begin{minipage}{0.48\textwidth}	
	\centering
	\includegraphics[width=1\columnwidth]{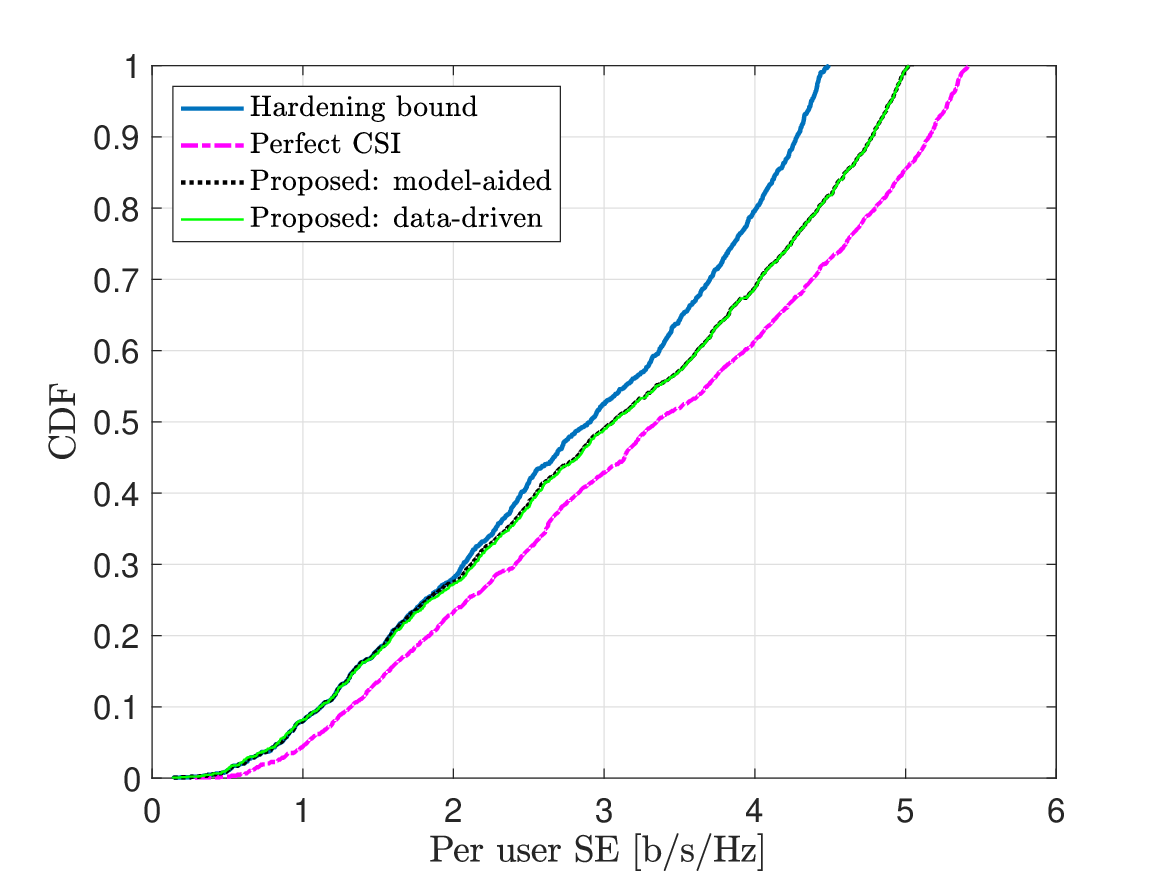}
	\caption{CDF of the SE per user for the uncorrelated channel model and MR precoding, with  DL SNR = $0\,$dB and $K = 3$.}
	\label{fig:SE_un_MR_K30}
	\end{minipage}
	\hfill
\begin{minipage}{0.48\textwidth}
	\centering
	\includegraphics[width=1\columnwidth]{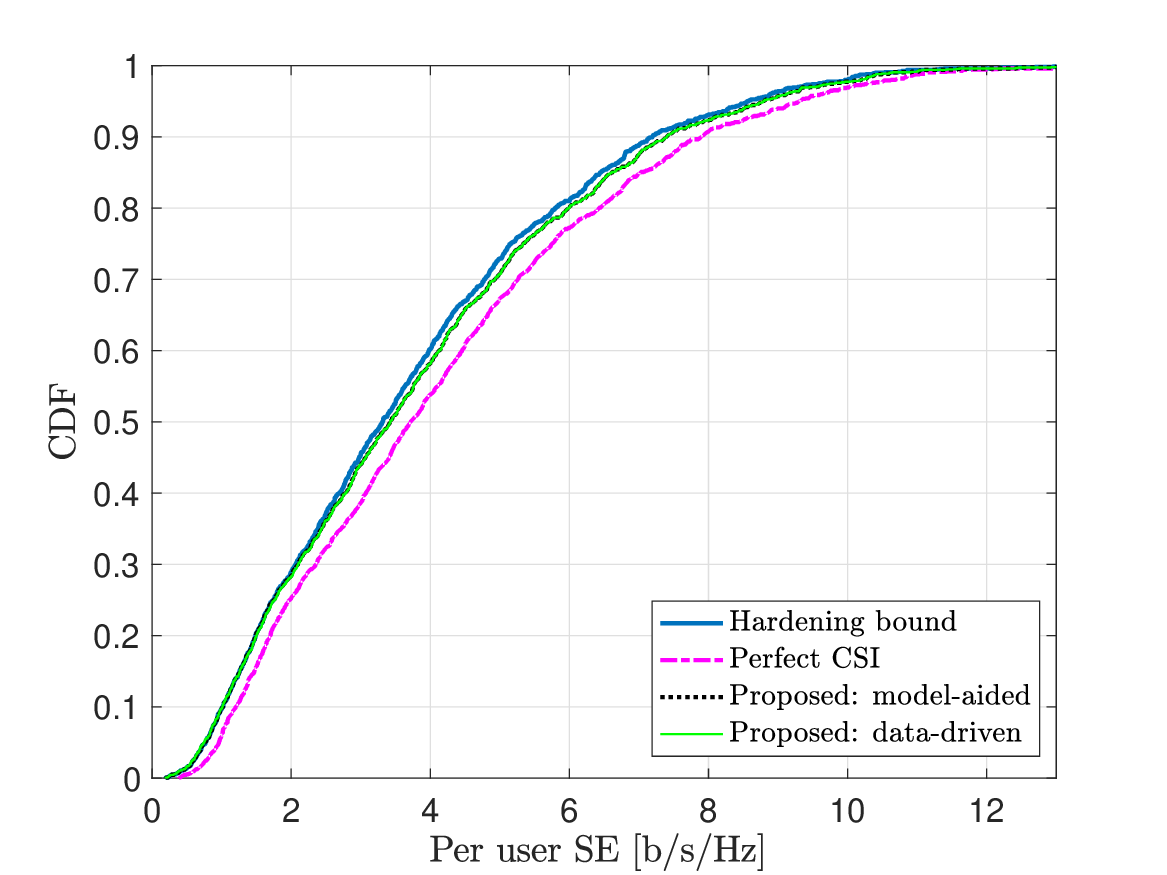}
	\caption{CDF of the SE per user for the uncorrelated channel model and ZF precoding, with DL SNR = $0\,$dB and $K = 3$.}
	\label{fig:SE_un_ZF_K30}
	\end{minipage}
\end{figure*}
{\color{black} To highlight the importance of the length of the coherence interval $\tau_c$, Fig.~\ref{fig:tau_c_length} provides the NMSE for an uncorrelated channel model with MR precoding at the BSs versus $\tau_c$. The perfect CSI and hardening bounds offer an NMSE independent of the coherence interval length $\tau_c$ since the channel estimates and channel statistics are estimated during the uplink pilot training phase with a fixed length $\tau_p$. It can be seen that the NMSE performance of the model-aided approach depends on the length of the coherence interval. Specifically, a longer coherence intervals yields a lower NMSE value and approaches the perfect CSI as a consequence of the law of large numbers when more and more data are collected.}

To provide more insights into the properties of the proposed estimators, Fig.~\ref{fig:mapping} shows one feature of the input data, namely $\xi'_{lk}$, versus the true value of $\alpha^{lk}_{lk}$. We consider one user (i.e., one large-scale realization) and $1000$ small-scale realization for MR precoding and $K=3$. While these data points are represented by 1000 circles, the lines represent the estimate of $\alpha^{lk}_{lk}$ that the two proposed estimators are providing for each value of $\xi'_{lk}$.
Note that the two estimators provide deterministic mappings while the true relationship is random since the data are random (as is always the case in estimation theory). We can see that the proposed estimators provide almost the same estimate when $\xi'_{lk}$ is large, while there is a gap between the estimators when $\xi'_{lk}$ is small. \textcolor{black}{The reason that the data-driven estimator obtains a better mapping than the model-based estimator 
is that it is trained for a system with finite $K$ and $\tau_c$, while the model-aided estimator is obtained from asymptotic arguments. Hence, the data-driven estimator provides lower NMSE and SE in the considered setup where $K$ is small.}

\textcolor{black}{ Fig.~\ref{fig:ASD} shows the impact of the spatial correlation level, which is expressed by the ASD. A small ASD represents high spatial correlation, and vice versa. The NMSE produced by the hardening bound is heavily depending on the strength of the spatial correlation. Specifically, the NMSE is quite high with a low ASD value (i.e., strong spatial correlation) since the channel vectors are less hardened. The approximation by the hardening bound becomes more accurate as the ASD increases (i.e., the channels approach spatially uncorrelated fading). In contrast, it is worth to notice that the proposed model-based approach is almost insensitive to the spatial correlation and close to the optimistic solution in the limiting regime when $\tau_c \rightarrow \infty$.}

\begin{figure*}[t]
	
\begin{minipage}{0.48\textwidth}
		\centering
		\includegraphics[width=1\columnwidth]{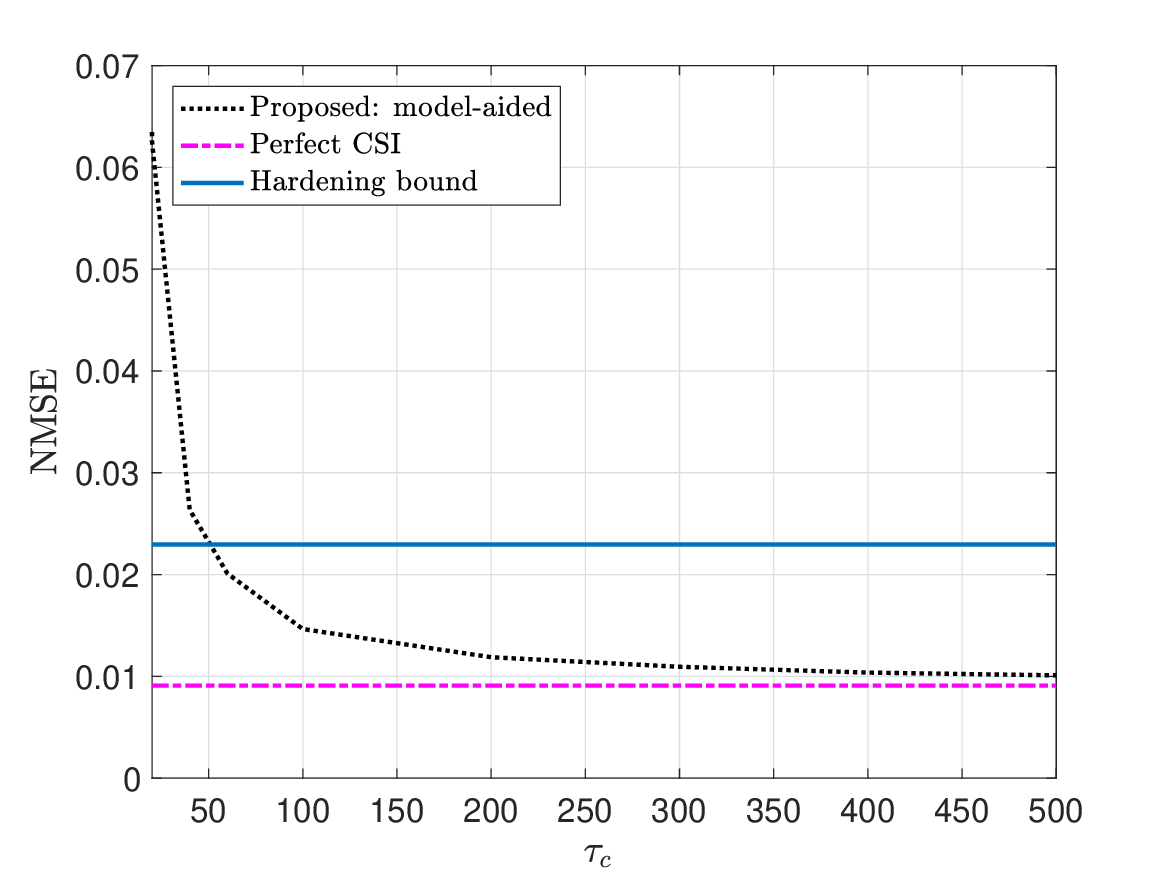}
		\caption{NMSE v.s. the coherence interval $\tau_c$ for uncorrelated channel model, MR precoding, DL SNR $=0\,$dB,  and $K =3$.}
		\label{fig:tau_c_length}
	\end{minipage}
	\hfill
	\begin{minipage}{0.48\textwidth}
	\centering
	\includegraphics[width=1\columnwidth]{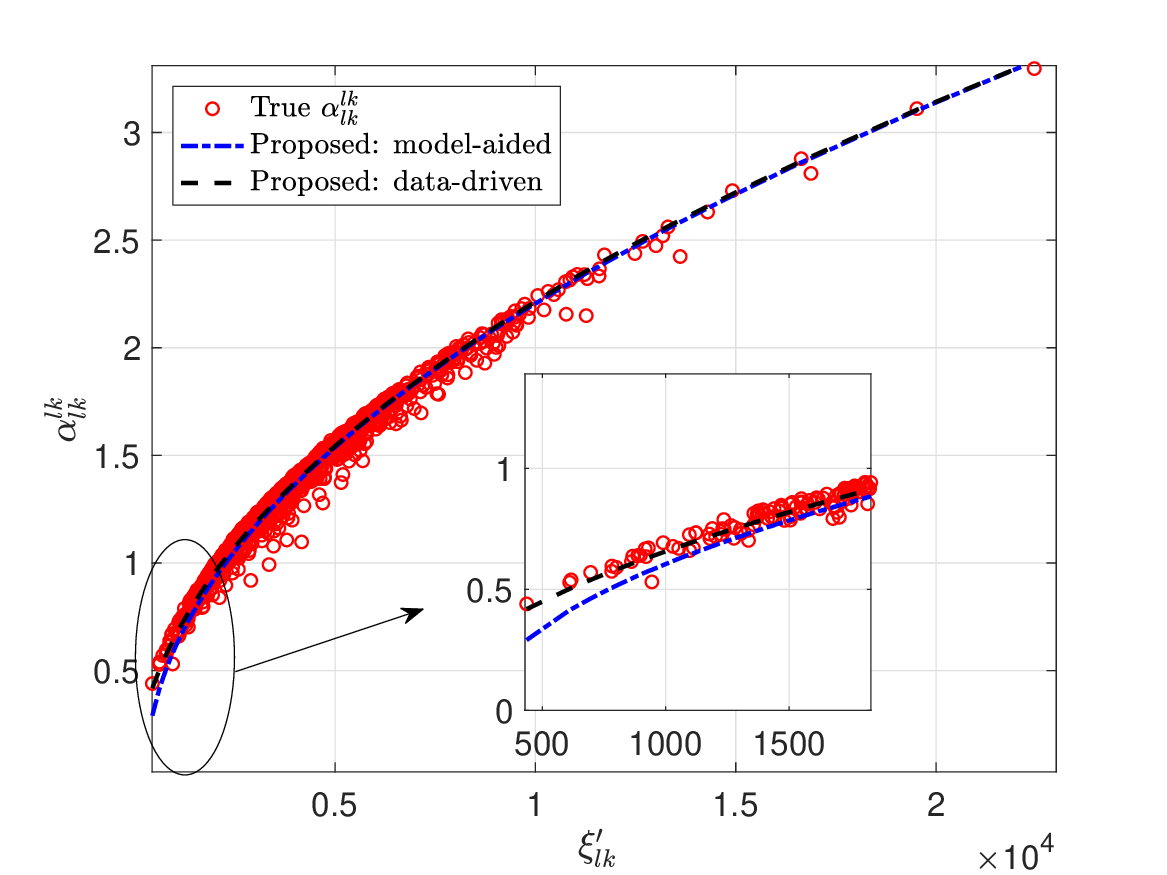}
	\caption{Representation of input $\xi'_{lk}$ vs true  $\alpha^{lk}_{lk}$ i.e. optimal output and estimates  of it from proposed estimators.}
	\label{fig:mapping}
	\end{minipage}
\end{figure*}

	\begin{figure}
	\centering
	\includegraphics[width=.6\columnwidth]{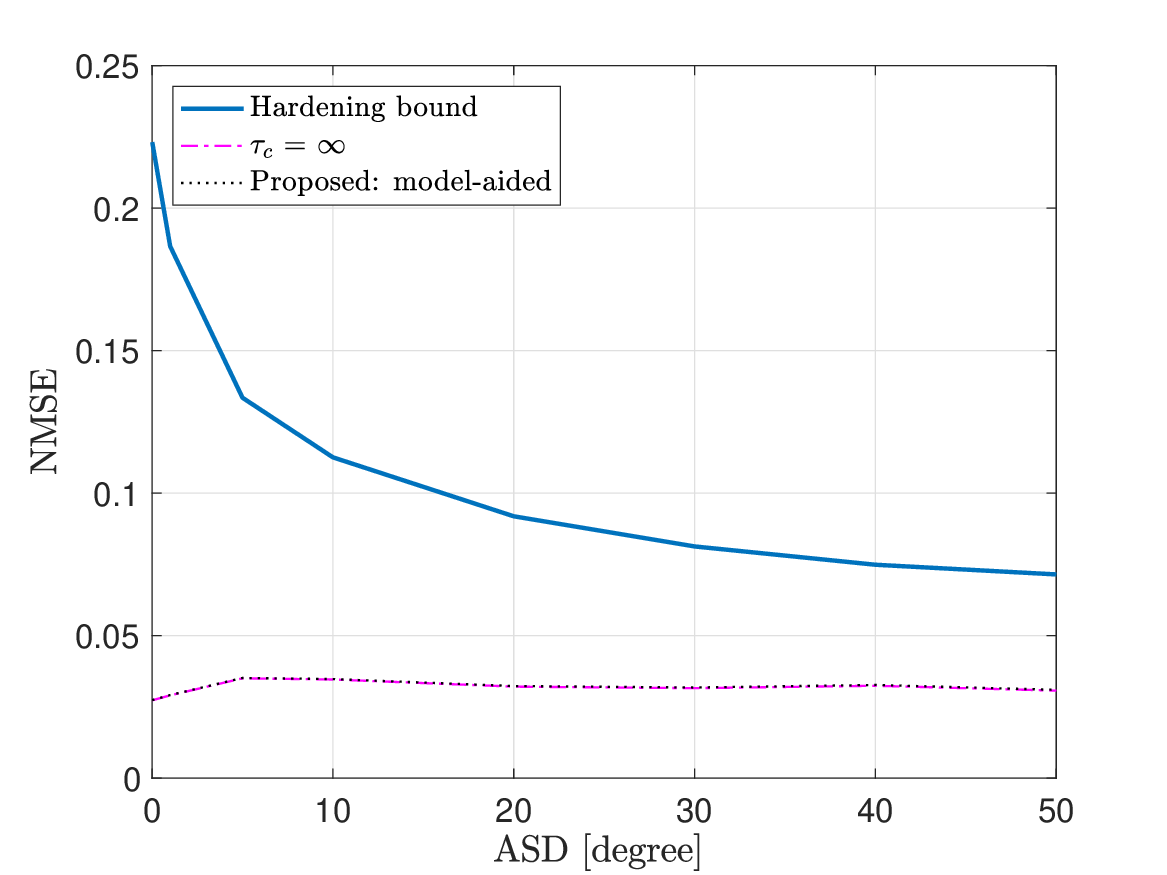}
	\caption{NMSE estimation for correlated channel model as a function of the ASD, MR precoding, DL SNR $=0\,$dB,  and $K =3$.}
	\label{fig:ASD}
	\end{figure}

	\section{Conclusion}\label{sec:conclusion}
This paper presented a model- and data-driven-based approach for downlink channel estimation in a multi-cell Massive MIMO system. The effective channel gains facilitate decoding of the desired signal from the received data signal at the users. The typical approach in Massive MIMO literature is to use the mean of effective channel gains as a piece of legitimate information for decoding. Using the mean of effective channel gains is a perfect assumption for hardening channels. If the channel hardening level is low, the performance fluctuation is high when using mean values. We investigate the performance for both having channel hardening and low level of channel hardening conditions.  We derived closed-form expressions for the downlink channel gain using MR and ZF precoding at BSs for uncorrelated Rayleigh fading channels for the model-based approach.
Furthermore, we provided a closed-form expression MR precoding for the correlated Rayleigh channel model. In addition, we proposed a second method that is data-driven and trains a neural network to identify a mapping between the available information and the effective channel gain.
Moreover, we present a performance comparison of the proposed model- and data-driven-based methods in terms of NMSE and per-user SE. The results highlight the superior performance of our proposed approaches in SE, specifically correlated fading, i.e., when the hardening level is low. 
	\appendix
	
	Several of the proofs in this appendix make uses of the weak law-of-large-numbers, which can be stated as follows.
	
	\begin{lemma} \label{lemma:LLN}
	Let $\{ X_n \}$ be a sequence of independent and identically distributed random variables with mean value $C$ and bounded variance \cite[Ch.~2]{ross2014introduction}. It then follows that
	\begin{equation}
	\frac{1}{N} \sum_{n=1}^{N} X_n \xrightarrow{P} C.
	\end{equation}
	\end{lemma}

	\subsection{Proof of Lemma~\ref{Lemma:asymptotic}} \label{Appendix:asymptotic}
	Let  $\tau_d = \tau_c - \tau_p$ denote the portion of each coherence interval dedicated to data transmission and notice that we consider the limit where $\tau_d \to \infty$. We can then rewrite \eqref{eq:Xilko} as
	\begin{equation} \label{eq:Xilkv1proof}
	\begin{aligned}
	&\xi_{lk}  = \frac{1}{\tau_d}\sum_{n=1}^{\tau_d} \left|y_{lk} [n]\right|^2 \\
	&=  \frac{1}{\tau_d} 	\sum_{n=1}^{\tau_d} \left|\underbrace{\sqrt{\eta_{lk}}\alpha_{lk}^{lk} s_{lk}[n] + \sum\limits_{\substack{k'=1, k' \neq k}}^K  \sqrt{\eta_{lk'}}\alpha_{lk}^{lk'} s_{lk'}[n]
		+ \sum\limits_{\substack{l'=1, l' \neq l}}^L \sum_{k'=1}^K \sqrt{\eta_{l'k'}}\alpha_{lk}^{l'k'} s_{l'k'}[n]}_{A_{n}} + \tilde{w}_{lk}[n]\right|^2\\
	& =  \frac{1}{\tau_d} 	\sum_{n=1}^{\tau_d}	\left|{A_{n}}\right|^2 + \frac{1}{\tau_d} 	\sum_{n=1}^{\tau_d}	A_{n}^{\ast}\tilde{w}_{lk}[n] +\frac{1}{\tau_d} 	\sum_{n=1}^{\tau_d}	{A_{n}}\tilde{w}^{\ast}_{lk}[n]+\frac{1}{\tau_d} 	\sum_{n=1}^{\tau_d}	\left|\tilde{w}_{lk}[n]\right|^2.		
	\end{aligned}
	\end{equation}
	As $\tau_d \to \infty$, it follows from Lemma~\ref{lemma:LLN} that the second and third terms in \eqref{eq:Xilkv1proof} converges to zero. Moreover, the fourth term converges to its mean value $\sigma_{\mathrm{DL}}^2$.
	Therefore, \eqref{eq:Xilkv1proof} simplifies asymptotically to 
	\begin{equation}\label{eq:simil}
	\begin{aligned} 
	\xi_{lk} & \asymp  
	 \frac{1}{\tau_d} 	\sum_{n=1}^{\tau_d}	\left|{A_{n}}\right|^2 +\sigma_{\mathrm{DL}}^2.
	\end{aligned}
	\end{equation}
	It remains to determine an asymptotic equivalent expression for $\frac{1}{\tau_d} 	\sum_{n=1}^{\tau_d}	\left|{A_{n}}\right|^2$. To this end, let us introduce a new variable $a_{lk}^{l'k'}[n] = \sqrt{\eta_{l'k'}}\alpha_{lk}^{l'k'} s_{l'k'}[n], \forall l,k,l',k'$ and expand $|A_n|^2$ as
	\begin{equation} \label{eq:Anv1}
	\begin{aligned}
	&	\left|{A_{n}}\right|^2 =    \left(a_{lk}^{lk}[n] + \sum\limits_{\substack{k'=1,\\ k' \neq k}}^K  a_{lk}^{lk'}[n]
	+ \sum\limits_{\substack{l'=1, \\l' \neq l}}^L \sum_{k'=1}^K a_{lk}^{l'k'}[n] \right)\left(a_{lk}^{lk}[n] + \sum\limits_{\substack{k'=1,\\ k' \neq k}}^K  a_{lk}^{lk'}[n]
	+ \sum\limits_{\substack{l'=1, \\l' \neq l}}^L \sum_{k'=1}^K a_{lk}^{l'k'}[n] \right)^\ast 
	\end{aligned}
	\end{equation}
		Let $\mathbb{E}_{s}$ denote the expectation with respect to the random data signals and noise (i.e., the conditional expectation given the channel realizations).
	Since all the cross-terms between different $a_{lk}^{l'k'}[n] $ have zero mean, it follows that 	
 
	\begin{equation}
	\begin{aligned}
		\mathbb{E}_{s} \left\{ \left|{A_{n}}\right|^2 \right\} &= \mathbb{E}_{s} \left\{ \left|a_{lk}^{lk}[n]\right|^2 \right\}  + \sum\limits_{\substack{k'=1, k' \neq k}}^K \mathbb{E}_{s} \left\{  \left| a_{lk}^{lk'}[n]\right|^2 \right\}+ \sum\limits_{\substack{l'=1, l' \neq l}}^L \sum_{k'=1}^K  \mathbb{E}_{s}\left\{ \left| a_{lk}^{l'k'}[n] \right|^2 \right\} \\
	& =  \eta_{lk}\left|\alpha_{lk}^{lk}\right|^2 +
	\sum\limits_{\substack{k'=1, k' \neq k}}^K  \eta_{lk'}\left|\alpha_{lk}^{lk'}\right|^2+  \sum\limits_{\substack{l'=1, l' \neq l}}^L \sum_{k'=1}^K  \eta_{l'k'}\left|\alpha_{lk}^{l'k'}\right|^2.
	\end{aligned}
	\end{equation}
	We can then apply  Lemma~\ref{lemma:LLN} a final time to \eqref{eq:simil} to obtain
		\begin{equation}\label{eq:simil2}
	\begin{aligned} 
	\xi_{lk} & \asymp  	\eta_{lk}\left|\alpha_{lk}^{lk}\right|^2 +
	\sum\limits_{\substack{k'=1, k' \neq k}}^K  \eta_{lk'}\left|\alpha_{lk}^{lk'}\right|^2+  \sum\limits_{\substack{l'=1, l' \neq l}}^L \sum_{k'=1}^K  \eta_{l'k'}\left|\alpha_{lk}^{l'k'}\right|^2 +\sigma_{\mathrm{DL}}^2,
	\end{aligned}
	\end{equation}
	which is the result stated in the lemma.

	\subsection{Proof of Lemma~\ref{lemma:asymptoicLLN}} \label{Appendix:asymptoicLLN}
	
			We first reformulate \eqref{eq:Xilk} by dividing both sides by $K$ to obtain an asymptotic equivalence
		\begin{equation}\label{eq:limitoriginal}
		\frac{1}{K} \xi_{lk} \asymp  \frac{1}{K}\eta_{lk}\left|\alpha_{lk}^{lk}\right|^2 + \underbrace{\frac{1}{K}\sum\limits_{\substack{k'=1,\\ k' \neq k}}^K  \eta_{lk'}\left|\alpha_{lk}^{lk'}\right|^2}_{\star} + \underbrace{\frac{1}{K}\sum\limits_{\substack{l'=1,\\ l' \neq l}}^L \sum_{k'=1}^K \eta_{l'k'}\left|\alpha_{lk}^{l'k'}\right|^2}_{\star\star} + \frac{1}{K} \sigma_{\mathrm{DL}}^2
		\end{equation} 
		as $\tau_c \to \infty$.
		At the right-hand side of \eqref{eq:limitoriginal}, $(\star)$ is a scaled-down version of the intra-cell interference and is further reformulated by adding and subtracting the mean value of $|a_{lk}^{lk'}|^2$ as
		\begin{equation}
		\frac{1}{K} \sum\limits_{\substack{k'=1,\\ k' \neq k}}^K  \eta_{lk'}\left|\alpha_{lk}^{lk'}\right|^2  = \frac{1}{K} \sum\limits_{\substack{k'=1,\\ k' \neq k}}^K  \eta_{lk'}\mathbb{E}\left\{\left|\alpha_{lk}^{lk'}\right|^2 \right\} + \frac{1}{K} \sum\limits_{\substack{k'=1,\\ k' \neq k}}^K  \eta_{lk'}\underbrace{\left(\left|\alpha_{lk}^{lk'}\right|^2 - \mathbb{E}\left\{\left|\alpha_{lk}^{lk'}\right|^2 \right\}\right)}_{B_{lk'}},
		\end{equation} 
		where $B_{lk'}$ has zero mean and comes from a common distribution with bounded variance. Hence, in the asymptotic regime where $K \rightarrow \infty$ (in the way stated in the lemma), $B_{lk'} \xrightarrow{P}  0 $. Note that the expectations are with respect to the channel realizations, while the convergence in probability considers both user locations and channel realizations.
		 It then follows from Lemma~\ref{lemma:LLN} that  
		\begin{equation}\label{eq:limit}
		\frac{1}{K} \sum\limits_{\substack{k'=1, k' \neq k}}^K  \eta_{lk'}\left|\alpha_{lk}^{lk'}\right|^2  - \frac{1}{K} \sum\limits_{\substack{k'=1, k' \neq k}}^K  \eta_{lk'}\mathbb{E} \left\{\left|\alpha_{lk}^{lk'}\right|^2 \right\} \xrightarrow{P}  0.
		\end{equation}
		In addition, $(\star\star)$ is a scaled-down version of inter-cell interference. By following a similar approach (with a fixed number $L$ of cells), $(\star\star)$ is reformulated as
		\begin{equation}
		\frac{1}{K} \sum\limits_{\substack{l'=1,\\ l' \neq l}}^L \sum_{k'=1}^K \eta_{l'k'}\left|\alpha_{lk}^{l'k'}\right|^2  = \frac{1}{K} \sum\limits_{\substack{l'=1,\\ l' \neq l}}^L \sum_{k'=1}^K \eta_{l'k'}\mathbb{E}\left\{\left|\alpha_{lk}^{l'k'}\right|^2 \right\} + \frac{1}{K} \sum\limits_{\substack{l'=1,\\ l' \neq l}}^L \sum_{k'=1}^K \eta_{l'k'}\underbrace{\left(\left|\alpha_{lk}^{l'k'}\right|^2 - \mathbb{E}\left\{\left|\alpha_{lk}^{l'k'}\right|^2 \right\} \right)}_{C_{l'k'}},
		\end{equation} 
		where $C_{l'k'}, \forall l' \neq l, \forall k',$ have zero mean and come from a common distribution with bounded variance. As $K \rightarrow \infty $ (in the way stated in the lemma), it follows that $C_{l'k'} \xrightarrow{P} 0$ which implies 
		\begin{equation}\label{eq:limit1}
		\frac{1}{K} \sum\limits_{\substack{l'=1, l' \neq l}}^L \sum_{k'=1}^K \eta_{l'k'}\left|\alpha_{lk}^{l'k'}\right|^2  - \frac{1}{K} \sum\limits_{\substack{l'=1, l' \neq l}}^L \sum_{k'=1}^K \eta_{l'k'}\mathbb{E}\left\{\left|\alpha_{lk}^{l'k'}\right|^2 \right\}  \xrightarrow{P}  0.
		\end{equation}		
		By using the asymptotic equivalences provided in \eqref{eq:limit} and \eqref{eq:limit1} in \eqref{eq:limitoriginal}, we provide the expression given in \eqref{eq:aymptotLemma}. 	
	\bibliographystyle{IEEEtran}
	\bibliography{refs}
\end{document}